\PassOptionsToPackage{quiet}{fontspec}
\documentclass[UTF8,oneside,12pt]{article}

\usepackage[nottoc]{tocbibind} 

\usepackage[a4paper,margin=1in]{geometry}
  \usepackage{float}

  \usepackage{makecell}
  \usepackage{xypic} \xyoption{all}

\usepackage{booktabs}
\usepackage{eso-pic}  

\usepackage{amsmath}        
\usepackage{amssymb}        
\usepackage{latexsym}       
\usepackage{mathrsfs}       
\usepackage{eucal}          
\usepackage{amsthm}         
\usepackage{mathtools}      
\usepackage{stmaryrd}       
\usepackage{esint}          
\usepackage{extarrows}      
\usepackage{enumerate}      
\usepackage{simpler-wick}   
\usepackage{xcolor}         
\usepackage{graphicx}       
\usepackage{tikz}         
\usepackage{tikz-cd}      
\usetikzlibrary{cd,matrix, calc, arrows,
decorations.pathreplacing, decorations.markings, decorations.pathmorphing,
backgrounds,fit,positioning,shapes.symbols,chains,shadings,fadings,intersections}
\usetikzlibrary{trees}
\usetikzlibrary{decorations.markings} 
\tikzset{double line with blue arrow/.style args={#1,#2}{decorate,decoration={markings,%
mark=at position 0 with {\coordinate (ta-base-1) at (0,1pt);
\coordinate (ta-base-2) at (0,-1pt);},
mark=at position 1 with {\draw[blue, #1] (ta-base-1) -- (0,1pt);
\draw[blue, #2] (ta-base-2) -- (0,-1pt);
}}}}
\usepackage{subcaption}     
\usepackage[colorlinks=true,linkcolor=blue]{hyperref}
\usepackage{cleveref}       
\usepackage{tcolorbox}      
\tcbuselibrary{most}        
\usepackage{bm}             
\usepackage{slashed}        

\usepackage{bbm}
\usepackage[T1]{fontenc}

\usepackage{tocloft} 

\newcommand{\ket}[1]{{\left| #1 \right>}}

\newcommand{\abracket}[1]{\left\langle#1\right\rangle}

\newcommand{\bracket}[1]{\left(#1\right)}

\renewcommand{\O}{\mathcal{O}}

\newcommand{\cA}{\mathcal A} 
\newcommand{\R}{\mathbb{R} }
\newcommand{\C}{\mathbb{C} }

\newcommand{\g}{{\mathfrak{g}}}

\newcommand{\pa}{\partial}

\newcommand{\OO}{{\mathcal O}}

\def\bC{\mathbb C}

\def\bH{\mathbb H}

\def\bK{\mathbb K}

\def\bP{\mathbb P}
\def\bQ{\mathbb Q}
\def\bR{\mathbb R}
\def\bZ{\mathbb Z}

\newcommand{\bi}{\begin{enumerate}}
\newcommand{\ei}{\end{enumerate}}
\newcommand{\be}{\begin{equation}}
\newcommand{\ee}{\end{equation}}
\newcommand{\bea}{\begin{equation*}\begin{aligned}}
\newcommand{\eea}{\end{aligned}\end{equation*}}
\newcommand{\beq}{\begin{eqnarray}}
\newcommand{\eeq}{\end{eqnarray}}
\newcommand{\ba}{\begin{array}}
\newcommand{\ea}{\end{array}}

\newcommand{\ols}[1]{\mskip.3\thinmuskip\overline{\mskip-.3\thinmuskip {#1} \mskip-.3\thinmuskip}\mskip.3\thinmuskip} 
\newcommand{\on}{\operatorname}

\newcommand{\nord}[1]{\,:\mathrel{#1}:\,}

\newcommand{\lb}{\left(}
\newcommand{\rb}{\right)}
\newcommand{\lcb}{\left\{}
\newcommand{\rcb}{\right\}}

\newcommand{\lsb}{\left[}
\newcommand{\rsb}{\right]}
\newcommand{\lan}{\left\langle}
\newcommand{\ran}{\right\rangle}

\newcommand{\ra}{\,\rightarrow\,}

\newcommand{\LRA}{\,\Leftrightarrow\,}

\def\Xint#1{\mathchoice
   {\XXint\displaystyle\textstyle{#1}}%
   {\XXint\textstyle\scriptstyle{#1}}%
   {\XXint\scriptstyle\scriptscriptstyle{#1}}%
   {\XXint\scriptscriptstyle\scriptscriptstyle{#1}}%
   \!\int}
\def\XXint#1#2#3{{\setbox0=\hbox{$#1{#2#3}{\int}$}
     \vcenter{\hbox{$#2#3$}}\kern-.5\wd0}}

\def\dint{\Xint-}

\def\dashint{\Xint-}

\def\blt{\bullet}

\def\hf{\frac{1}{2}}
\def\p{\partial}

\def\pb{\ols \p}

\def\sym{\on{Sym}}
\def\asym{\scalebox{.9}[1]{$\bigwedge$}}

\def\gsp{\on{Sp}}

\def\bC{\mathbb C}

\def\bH{\mathbb H}

\def\bK{\mathbb K}

\def\bP{\mathbb P}
\def\bQ{\mathbb Q}
\def\bR{\mathbb R}
\def\bZ{\mathbb Z}

\def\fg{\mathfrak{g}}
\def\fh{\mathfrak{h}}
\def\sp{\mathfrak{sp}}

\def\sO{\mathscr O}

\def\cA{\mathcal A}

\def\cD{\mathcal D}
\def\cH{\mathcal H}

\def\cL{\mathcal L}
\def\cM{\mathcal M}

\def\cO{\mathcal O}

\def\cV{\mathcal V}
\def\cW{\mathcal W}

\newcommand{\E}{{\mathbb E}}   

\newcommand{\cE}{{\mathcal E}}

\renewcommand{\L}{\mathscr{L} }

\newcommand{\ut}[1]{$\underline{\text{#1}}$}

\newcommand{\lr}[1]{\left\{ #1\right \}}

\renewcommand{\a}{\alpha}
\renewcommand{\b}{\beta}

\newcommand{\e}{\varepsilon }

\newcommand{\inv}{^{-1}}

\newcommand{\w}{\wedge}

\newcommand{\lra}{ \longrightarrow}

\newcommand{\tc}[1]{{\raisebox{.5pt}{\textcircled{\raisebox{-.9pt}{#1}}}}}

\newcommand{\EE}{{\mathcal{E}}}

\newcommand{\Lie}{\text{Lie}}

\newcommand{\blue}[1]{{\color{blue}{#1}}}

\DeclareMathOperator{\Hom}{Hom}

\DeclareMathOperator{\Map}{Map}

\DeclareMathOperator{\Sym}{Sym}
\DeclareMathOperator{\Tr}{Tr}

\let\Im\undefined

\DeclareMathOperator{\Im}{Im}

\theoremstyle{plain}
\newtheorem{thm}{Theorem}[section]
\newtheorem{thm-defn}{Theorem/Definition}[section]

\newtheorem{lem-defn}[thm]{Lemma/Definition}
\newtheorem{prop}[thm]{Proposition}
\newtheorem{cor}[thm]{Corollary}
\newtheorem*{conj}{Conjecture}

\newtheorem{eg}[thm]{Example}

\theoremstyle{definition}
\newtheorem{defn}[thm]{Definition}

\theoremstyle{remark}
\newtheorem{rmk}[thm]{Remark}

\begin{document}

\newcommand{\myauthor}{Si Li} 

 \title{Quantization and Algebraic Index}
 \author{\myauthor} 
  \date{}

  \maketitle
{\centering Dedicated to Prof. S.-T. Yau on the occasion of his 75th Birthday \par}

\abstract{This article reviews the program on connecting Batalin-Vilkovisky (BV) quantization with index theories of algebraic type. We explain how the classical algebraic index theorem can be proved in terms of BV quantization of topological quantum mechanics. This is generalized to 2d chiral CFT in which we present an elliptic chiral analog of the algebraic index theory. As an application, we show how the generating function of all genus Gromov-Witten invariants on elliptic curves is mirror equivalent to an elliptic chiral index in the mirror BCOV theory. 
}

\tableofcontents

\section{Introduction}

\begin{center}
\tikzset{every picture/.style={line width=0.75pt}} 

\begin{tikzpicture}[x=0.75pt,y=0.75pt,yscale=-1,xscale=1,scale=0.8]

\draw   (218,64.5) .. controls (218,47.1) and (265.12,33) .. (323.25,33) .. controls (381.38,33) and (428.5,47.1) .. (428.5,64.5) .. controls (428.5,81.9) and (381.38,96) .. (323.25,96) .. controls (265.12,96) and (218,81.9) .. (218,64.5) -- cycle ;
\draw   (66,197.5) .. controls (66,180.1) and (113.12,166) .. (171.25,166) .. controls (229.38,166) and (276.5,180.1) .. (276.5,197.5) .. controls (276.5,214.9) and (229.38,229) .. (171.25,229) .. controls (113.12,229) and (66,214.9) .. (66,197.5) -- cycle ;
\draw   (371,196.5) .. controls (371,179.1) and (418.12,165) .. (476.25,165) .. controls (534.38,165) and (581.5,179.1) .. (581.5,196.5) .. controls (581.5,213.9) and (534.38,228) .. (476.25,228) .. controls (418.12,228) and (371,213.9) .. (371,196.5) -- cycle ;
\draw   (191.15,162.9) -- (208.46,139.71) -- (211,143.23) -- (255.78,110.92) -- (253.24,107.4) -- (280.7,98.28) -- (263.39,121.47) -- (260.85,117.95) -- (216.07,150.26) -- (218.61,153.78) -- cycle ;
\draw   (362.14,98.72) -- (389.88,106.97) -- (387.46,110.56) -- (433.24,141.43) -- (435.66,137.83) -- (453.71,160.46) -- (425.97,152.21) -- (428.39,148.62) -- (382.61,117.75) -- (380.19,121.35) -- cycle ;
\draw   (283.53,197.03) -- (303.59,187.99) -- (303.62,192.45) -- (343.85,192.18) -- (343.82,187.72) -- (364,196.5) -- (343.94,205.54) -- (343.91,201.09) -- (303.68,201.35) -- (303.71,205.81) -- cycle ;

\draw (230,53) node [anchor=north west][inner sep=0.75pt]   [align=left] {\blue{Quantum Field Theory}};
\draw (85,186) node [anchor=north west][inner sep=0.75pt]   [align=left] {\blue{Homological Algebra}};
\draw (421,187) node [anchor=north west][inner sep=0.75pt]   [align=left] {\blue{Index Theory}};

\end{tikzpicture}

 \end{center}

It is well known that the Atiyah-Singer index theorem is closely related to supersymmetric/topological quantum mechanics \cite{Atiyah-circular,getzler1983pseudodifferential,windey1983supersymmetric,witten1982constraints}. Though not rigorous, this physics interpretation provides a clear and deep insight into the origin of index theorem via the geometry of the loop space. As a natural generalization,  one can replace the loop by the two-dimensional torus. This leads to Witten's proposal  \cite{witten1987elliptic,witten1988index} for the index of Dirac operators on loop spaces.

In \cite{Fedosov-book,Nest-Tsygan}, Fedosov and Nest-Tsygan established the algebraic index theorem for deformation quantized algebras as the algebraic analogue of Atiyah-Singer index theorem. It was further shown  \cite{nest1996formal} that the original Atiyah-Singer index theorem can be deduced from this algebraic one. In \cite{GLL-BV, GLX-Effective}, we established an exact connection between the algebraic index theorem and topological quantum mechanics via a trace map constructed in the Batalin-Vilkovisky(BV) formalism \cite{BATALIN198127}. This sets up a mathematical understanding of the physics approach to index theorem in terms of an exact low-energy effective quantum field theory \cite{GLX-Effective}. Such connection between quantization and algberaic index can be naturally extended to quantum field theory on other geometric objects, such as the torus. In \cite{LS-VOA,Gui:2021dci}, we developed the effective BV quantization theory of two-dimensional chiral theory and established a chiral analogue of the algebraic index theory on the torus. 

This paper reviews the program on BV quantization and  index theories of algebraic type developed in \cite{GLL-BV, GLX-Effective,LS-VOA,Gui:2021dci}. Here we summarize the main structures. 

 Let us denote by $\mathbf{k}$ the field of Laurent series $\mathbb{C}((\hbar))$. Roughly speaking, BV quantization in quantum field theory on $X$ leads to the following data (we will give more details in the body of the text)
  \begin{enumerate}
    \item A factorization algebra  of local observables (we follow the set-up in \cite{costello2021factorization}).
 $$
 \mathrm{Obs}:\quad \text{a $\mathbf{k}$-module equipped with certain algebra structure.}
 $$
 It carries an algebraic structure called factorization product (or operator product expansion in physics terminology).
    \item  A (factorization) chain complex
 $$
 C_{\bullet}(\mathrm{Obs}):\quad \text{a $\mathbf{k}$-chain complex}, \qquad d: \quad \text{the differential}.
 $$
It captures the algebraic structure and global information from local observables.
\item A BV algebra $(O_{\mathrm{BV}}, \Delta)$
  $$
  O_{\mathrm{BV}}:\quad \text{a BV algebra over}\ \mathbb{C}, \qquad \text{$\Delta:\quad$ the BV operator}
  $$
  together with a BV integration map
   $$
 \int_{\mathrm{BV}}: O_{\mathrm{BV}}\to \mathbb{C}, \qquad \text{such that} \qquad \int_{\mathrm{BV}} \Delta(-)=0.
 $$
 In physics, $O_{\mathrm{BV}}$ are functions on the space of zero modes at low energy. $\int_{\mathrm{BV}}$ is a choice (related to the gauge fixing) of the integration map on zero modes. It will be $\hbar$-linearly extended when the quantum parameter $\hbar$ is involved. 
 \item A $\mathbf{k}$-linear map (encoding the path integral in physics)
 $$
 \mathrm{\mathbf{Tr}}:C_{\bullet}(\mathrm{Obs})\rightarrow O_{\mathrm{BV},\mathbf{k}}=O_{\mathrm{BV}}\otimes_{\mathbb{C}}\mathbb{C}((\hbar))
 $$
  satisfying the \textbf{quantum master equation} (QME)
  $$
  (d+\hbar\Delta) \mathrm{\mathbf{Tr}}=0.
  $$
 In other words, QME says $\mathbf{Tr}$ is a chain map intertwining $d$ and $-\hbar \Delta$. In physics, it describes the quantum gauge consistency condition in terms of BV formalism. Index is obtained as the partition function of the model, which can be formulated as
$$
\text{Index}= \int_{\mathrm{BV}}  \mathrm{\mathbf{Tr}}(1).
$$
  \end{enumerate}

In Section \ref{sec:BV-quantization}, we review the theory of effective BV quantization. In Section \ref{sec:TQM}, we explain the 1d example of topological quantum mechanics on the circle and show how the above structures lead to the  algebraic index theorem. In this case
\begin{itemize}
    \item The factorization algebra is the Weyl algebra: $\on{Obs}=\cW_{2n}$.
    \item The factorization complex is the Hochschild chain complex $\lb C_\blt(\cW_{2n}),b\rb$.
    \item BV algebra on zero modes: $\lb\cA,\Delta\rb=\lb \Omega^\blt(\bR^{2n}),\cL_{\omega^{-1}}\rb.$
    \item Free correlation map 
    \bea \lan -\ran: C_{\blt}(\cW_{2n})\to \Omega^\blt(\bR^{2n})((\hbar)), \quad b\mapsto \hbar \cL_{\omega^{-1}}.\eea
    \item $\on{Index}=\int_{BV}\lan 1\ran=\lsb e^{\omega_\hbar/\hbar}\widehat{A}\rsb$.
\end{itemize}

In Section \ref{sec:2d1}, we explain the 2d chiral example and the elliptic chiral analogue of algebra index via $\beta\gamma-bc$ system. In this case, the factorization complex is the chiral chain complex of the corresponding vertex operator algebra. The trace map arising from BV quantization on elliptic curves will be called the \emph{elliptic trace map}. 

The above two examples in Section \ref{sec:BV-quantization} and \ref{sec:2d1} share a special property: they are both UV finite theories. A conjectured structure for BV quantization of general UV finite theory is presented in Section \ref{sec:UV}.

\begin{table}[h]
  \centering
{
\begin{tabular}{|c|c|}
\hline
 1d TQM & 2d Chiral QFT\\
\hline
 \makecell[c]{Associative algebra} & \makecell[c]{Vertex operator algebra} \\
\hline
 {\color{red}Hochschild homology} & {\color{blue}Chiral homology} \\
\hline
 \makecell[c]{BV QME:\\$\displaystyle {\color{red}(\hbar \Delta +b) \langle -\rangle _{1d} =0}$} & \makecell[c]{BV QME:\\$\displaystyle {\color{blue}(\hbar \Delta +d_{ch}) \langle -\rangle _{2d} =0}$} \\
\hline
 \makecell[c]{$\displaystyle \langle \mathcal{O}_{1} \otimes \cdots \otimes \mathcal{O}_{n} \rangle _{1d} =$integrals\\ on the {\color{red}compactified } \\ configuration spaces of  $\displaystyle S^{1}$} & \makecell[c]{$\displaystyle \langle \mathcal{O}_{1} \otimes \cdots \otimes \mathcal{O}_{n} \rangle _{2d} =${\color{blue}regularized} \\{\color{blue}integrals} of singular forms on $\displaystyle \Sigma ^{n}$} \\
\hline
 Algebraic Index & Elliptic Chiral Algebraic Index \\
 \hline
\end{tabular}}
  \end{table}

\noindent \textbf{Acknowledgement:} The author would like to thank Prof. S.-T. Yau for his invaluable support and encouragement in my career. This work is supported by the National Key Research and Development Program of China  (NO. 2020YFA0713000).

\section{Effective Theory of BV Quantization}\label{sec:BV-quantization}

In this section, we review Costello's homotopic theory of effective BV quantization \cite{costello2011renormalization}. This is the basic framework that we will use to establish  the connection between BV quantization and index theories of algebraic type. We follow the presentation in \cite{LS-VOA}. 

\subsection{(-1)-shifted Symplectic Structure}
We first explain that classical field theories and their quantizations have a  universal description in terms of (-1)-shifted symplectic structure. This is particularly convenient to quantize gauge theories in the BV framework.

We start with the finite dimensional toy model. Let $(V,Q,\omega)$ be a finite dimensional dg (differential graded) symplectic space. Here
\begin{itemize}
    \item V is a finite dimensional graded vector space.
    \item $Q:V \longrightarrow V$ differential, $\deg Q=1$ and $Q^2=0$.
      \item $\omega:\w^2 V\longrightarrow\R$ non-degenerate pairing of {deg=$-1$}, that is, 
      $$
      \omega(a,b)=0, \text{ unless } |a|+|b|=1.
      $$
      \item $\omega$ is Q-compatible $ Q(\omega)=0$, i.e., 
      $$
    \omega(Q(a),b) +(-1)^{|a|}\omega(a,Q(b))=0. 
      $$
\end{itemize}
The non-degeneracy of $\omega$ leads to linear isomorphisms
\begin{align*}
       \omega: V^\vee &\overset{\sim}{\longrightarrow} V[1]\\
    \Longrightarrow \quad \w^2(V^\vee)&\overset{\sim}{\longrightarrow} \w^2(V[1]) \simeq Sym^2(V)[2]\\
    \omega &\longleftrightarrow K[2]
\end{align*}
  
Here $K=\omega^{-1}\in \Sym^2(V)$ is the Poisson Kernel and
$$
\deg(K)=1,~~Q(K)=0. 
$$
We obtain a triple $(A,Q,\Delta)$ as follows
\begin{itemize}
    \item $A=\O(V):=\widehat{\Sym}(V^\vee)$ (formal power series on V)
    \item $Q:A\lra A$ derivation induced dually from $Q:V\lra V$
    \item BV operator
    $$
    \Delta=\Delta_K:A\lra A
    $$
    by contracting with the Poisson Kernel $K$
    \begin{align*}
        \Delta_K: \Sym^m(V)\lra \Sym^{m-2} (V). 
\end{align*}
       Explicitly, for $\a_i\in V^\vee$
       \begin{align*}
       \Delta_K(\a_1\otimes\cdots\otimes\a_m)
         = \sum_{i<j}\pm \langle K,\a_i\otimes\a_j \rangle\a_1\otimes\cdots\widehat{\a_i}\otimes\cdots\widehat{\a_i}\otimes\cdots \otimes\a_m. 
    \end{align*}
    \item $\Delta_K$ induces a BV bracket on $A$ by
    $$
    \{a,b\}:=\Delta_K(ab)-(\Delta_K a)b-(-1)^{|a|}a \Delta_kb
    $$
    Here $|a|$ is the degree of $a$. 
    \item Since $K$ is $Q-$closed, we have
    $$
    [Q,\Delta_K]:=Q\Delta_K+\Delta_K Q=0
    $$
\end{itemize}
The triple $(A,Q,\Delta)$ is exactly the data of a DGBV. Given such a DGBV, we can talk about
\begin{itemize}
\item \underline{Classical master equation}: 
\begin{align*}
    QI_0+\frac{1}{2}\lr{I_0,I_0}=0\qquad \text{for}\quad I_0\in A, \ \deg(I_0)=0. 
\end{align*}
Then the classical BRST operator $\delta=Q+\lr{I_0,-}$ satisfies $\delta^2=0$. 
\item \underline{Quantum master equation}:
\begin{align*}
    QI+\hbar I+\frac{1}{2}\lr{I,I}=0
    \Longleftrightarrow(Q+\hbar \Delta)e^{I/\hbar}=0, \quad \text{for}\quad    I=I_0+\hbar I_1+...\in A[[\hbar]]. 
\end{align*}
Then the quantum BRST operator $\delta^\hbar=Q+\hbar\Delta+\lr{I,-}$ satisfies
$
(\delta^\hbar)^2=0.
$
\end{itemize}

\subsubsection*{Classical Field Theory}
Now we discuss the QFT situation. For our purpose, we focus on theories where fields are sections of vector bundles.  A classical field theory can be organized into $\infty-$dimensional (-1)-shifted dg symplectic space
    $$
    (\EE,Q,\omega)
    $$
    \begin{itemize}
        \item $\EE=\Gamma(X,E^\blt)$ the space of fields. Here $E^\blt$ is a graded vector bundle on $X$.
        \item $(\EE,Q)$ elliptic complex
\[\begin{tikzcd}
	{...} & {\EE^{-1}} & {\EE^0} & {\EE^1} & {...}
	\arrow[from=1-1, to=1-2]
	\arrow["Q", from=1-2, to=1-3]
	\arrow["Q", from=1-3, to=1-4]
	\arrow["Q", from=1-4, to=1-5]
\end{tikzcd}\]
For example, $Q=\bar{\partial}$ or $d$. 
    \item $\omega$: {local} (-1)-symplectic pairing
$$
\omega(\a,\b)=\int_X\langle\a,\b\rangle, ~~\forall\a,\b\in\EE
$$
and compatible with Q.
    \end{itemize}

\begin{eg}[Chern-Simons Theory]
Let $X$ be a $\dim=3$ manifold, and $\g$ be a Lie algebra with trace pairing $Tr:\g\otimes\g\lra \R$. The space of fields is 
$$
\EE=\Omega^\blt (X,\g)[1].
$$
The degree shifting $[1]$ gives the following intepretation. 
\begin{center}
\begin{tabular}{|c|c|c|c|c|} 
  \hline 
  & $\Omega^0 (X,\g)$ & $\Omega^1 (X,\g)$ & $\Omega^2(X,\g)$ & $\Omega^3 (X,\g)$ \\ 
  \hline 
deg & -1 & 0 & 1 & 2\\
\hline
& $c$ & $A$ & $A^\vee$ & $c^\vee$\\
\hline
& ghost & field & anti-field & anti-ghost \\
  \hline 
\end{tabular}
\end{center}

$Q=d$ is the de Rham differential. The $(-1)$-symplectic pairing is
    $$
    \omega(\a,\b)=\int_XTr(\a\w\b),\quad \a,\b\in\EE
    $$
which pairs $0$-forms with $3$-forms and pairs $1$-forms with $2$-forms. 
\end{eg}
\begin{eg}[Scalar Field Theory in BV formalism] The field complex $\EE$ is 
\begin{center}
\begin{tikzcd}
C^\infty(M) \ar[r,"Q=\Delta+m^2"] & C^\infty(M)\\
\deg=0  & \deg=1 \\
\phi & \phi^\vee
\end{tikzcd}
\end{center}
The $(-1)$-symplectic pairing is $$
\omega(\phi,\phi^\vee)=\int_M\phi\phi^\vee.
$$
\end{eg}

\subsubsection*{UV Problem}

Let us now perform the same construction of DGBV algebra following the toy model. We first need the notion of "functions" $\O(V)=\widehat{\Sym}(V^\vee)$ on $V$.
\begin{itemize}
    \item \underline{linear function}: we have to take a continuous linear dual and so
    $$
    \EE^\vee =\Hom_X(\EE,\R)
    $$ 
    is given by distributions.
    \item $(\EE^\vee)^{\otimes n}=\Hom_{X\times...\times X}(\EE^{\otimes n},\R)$ are distributions on $X^n$. Here $$
    \EE^{\otimes n}=\Gamma(X^n,E^{\boxtimes n})
    $$ 
    is the completed tensor product. Thus 
    $$
    \Sym^m (\EE^\vee):=(\EE^\vee)^{\otimes m}/S_m
    $$
    is well-defined by distributions on $X^m$. As a result, we can form
    $$
    \O(\EE)=\prod_{m\ge 0}\Sym^m(\EE^\vee)
    $$ 
    representing (formal) functions on $\EE$.
    \item $Q:\EE\lra\EE$ induces duality $Q:\EE^\vee\lra\EE^\vee$ on distributions, and gives rise to 
    $$
    Q:\O(\EE)\lra\O(\EE).
    $$
    \item \underline{BV operator}: Let 
    $
    K=\omega^{-1}$ be the Poisson kernel as above. Since 
    $$
    \omega=\int\langle-,-\rangle
    $$
    is an integral, its inverse $K$ is a $\delta$-function distribution supported on the diagonal of $X\times X$. Thus
    $K$ is \underline{NOT} a smooth element in $\Sym^2(\EE)$, but a distributional section. As a result, the naive BV operator 
    $$
    \Delta_K: \Sym^m(\EE^\vee)\to \Sym^{m-2}(\EE^\vee)
    $$
    is \underline{ill-defined} since we can not pair two distributions. This is essentially the {Ultra-Violet problem}. Renormalization is needed in the quantum theory!
\end{itemize}

Before we move on to discuss the issue of renormalization, let us point out that the classical theory is actually well-behaved. Let $\O_{loc}(\EE)\subset\O(\EE)$ denote the subspace of local functionals, i.e., those by integrals of lagrangian densities
$$\O_{loc}(\EE)=\lr{\int_X\L(...)}
$$
Although the BV operator $\Delta_K$ is ill-defined, the associated BV bracket $\lr{-,-}$ is actually well-defined on local functionals since $\delta$-function can be integrated.
$$
\lr{-,-}:\O_{loc}(\EE)\otimes\O_{loc}(\EE)\lra\O_{loc}(\EE)
$$
\[
\begin{tikzpicture}
    \coordinate (a) at (0,0);
    \coordinate (b) at (2,0);
    \coordinate (a1) at (-0.8,0.5);
    \coordinate (a2) at (-1,0);
    \coordinate (a3) at (-0.8,-0.5);
    \coordinate (b1) at (2.5,0.5);
    \coordinate (b3) at (2.5,-0.5);

    \draw[thick] (a) -- (a1);
    \draw[thick] (a) -- (a2);
    \draw[thick] (a) -- (a3);
    \draw[thick] (b) -- (b1);

    \draw[thick] (b) -- (b3);
    \draw[thick] (a) -- (b);

    \node at (1,0.25) {$K\sim \delta$};
    \node at (3,0) {{\Huge=}};
    \node at (4.5,0) {\large$\int_X(-)$};
    \node at (0,-1) {$\int_X\L_1$};
    \node at (2,-1) {$\int_X\L_2$};
\end{tikzpicture}
\]
In other words, 
\begin{itemize}
\item CME makes sense for local functionals
\item QME needs renormalization
\end{itemize}
We refer to \cite{costello2011renormalization} for detailed discussions on this issue. 

\begin{eg}[Chern-Simons theory]
$\EE=\Omega^\blt (X,\g)[1]$
\begin{center}
\begin{tabular}{ccccc} 
  & $\Omega^0 (X,\g)$ & $\Omega^1 (X,\g)$ & $\Omega^2(X,\g)$ & $\Omega^3 (X,\g)$ \\ 
deg & -1 & 0 & 1 & 2\\
& $c$ & $A$ & $A^\vee$ & $c^\vee$\\
& ghost & field & anti-field & anti-ghost \\
\end{tabular}
\end{center}
Let $\cA=C+A+A^\vee+C^\vee\in \EE$ denote the master field collecting all components. Then the BV Chern-Simons action is
$$
CS[\cA]=\int_XTr(\hf \cA\w d\cA +\frac{1}{6}\cA\w[\cA,\cA]).
$$
This takes the same form as ordinary Chern-Simons except that we have expanded $\cA$ to get terms containing different components. The first quadratic term is denoted by $S_{free}$, the free part.  The second cubic term is denoted by $I$, the interaction part.

    $CS$ satisfies the following classical master equation 
    $$
    \lr{CS,CS}=0.
    $$
This follows from the general argument that classical gauge theory is organized into a solution of classical master equation. Let us separate the free part and interaction
$$
CS=S_{free}+I.
$$

It is easy to see that 
$$
\lr{S_{free},-}=d\ (=Q)
$$
which corresponds to the de Rham differential. Thus
\begin{align*}
    &\lr{CS,CS}=0\\
    \Leftrightarrow &\quad  \hf \lr{S_{free},S_{free}}+\lr{S_{free},I}+\hf\lr{I,I}=0\\
    \Leftrightarrow & \quad QI+\hf \lr{I,I}=0
\end{align*}
This is precisely the form of classical master equation in our DGBV.
\end{eg}

\subsection{Effective Renormalization}
Assume we have a classical field theory $(\EE=\Gamma(X,E^\blt),Q,\omega)$ with classical local functional $I_0$ (interaction) satisfying CME
$$
QI_0+\hf \lr{I_0,I_0}=0. 
$$
As we explained before, quantization asks for
$$
I_0\lra I=I_0+\hbar I_1+\hbar I_2+...\in\O(\EE)[[\hbar]]
$$
satisfying QME
$$
``QI_0+\hf \lr{I_0,I_0}+\hbar\Delta I=0".
$$
\ut{Problem:} $\Delta I$ is \ut{NOT} well-defined. In the following, we explain Costello's homotopic renormalization theory to solve this problem.

\subsubsection*{Toy Model}
To motivate the construction, let us look back again at the toy model where $(V,Q,\omega)$ is finite dimensional (-1)-shifted dg symplectic space. The Poisson kernel 
$$
K_0\in \Sym^2(V)
$$
has $\deg(K_0)=1$ and satisfies $Q(K_0)=0$. This allows us to construction the BV operator $\Delta_0$ by contracting with $K_0$ and obtain the DGBV triple $(A,Q,\Delta_0)$. 

Let us now consider the change of $K_0$ by chain homotopy. Let 
$$
P\in \Sym^2(V),\quad \deg(P)=0. 
$$
Define
$$
K_P=K_0+Q(P)=K_0+(Q\otimes 1+1\otimes Q) P.
$$
We again have
\begin{itemize}
    \item  $K_P\in \Sym^2(V),\deg(K_P)=1$
    \item $Q(K_P)=0$
\end{itemize}
Thus we can construct a new BV operator 
$$
\Delta_P=\text{contraction with } K_P 
$$
such that $(\O(V),Q,\Delta_P)$ forms a new {DGBV}. 

To see the relation with the original DGBV, denote
$$
\p_P:\Sym^m(V^\vee)\lra \Sym^{m-2}(V^\vee)
$$
where $\p_P$ is a 2nd order operator of contracting with $P\in \Sym^2(V)$
\begin{prop}
    The following diagram commutes
    \[\begin{tikzcd}
	{\O(V)[[\hbar]]} && {\O(V)[[\hbar]]} \\
	\\
	{\O(V)[[\hbar]]} && {\O(V)[[\hbar]]}
	\arrow["{e^{\hbar \p_P}}", from=1-1, to=1-3]
	\arrow["{Q+\hbar \Delta_0}", from=1-1, to=3-1]
	\arrow["{Q+\hbar \Delta_P}", from=1-3, to=3-3]
	\arrow["{e^{\hbar \p_P}}", from=3-1, to=3-3]
\end{tikzcd}\]
i.e.
$$
(Q+\hbar \Delta_P)e^{\hbar \p_P}=e^{\hbar \p_P}(Q+\hbar \Delta_0).
$$
\end{prop}
\begin{proof} This follows from the chain homotopy relation $K_P=K_0+Q(P)$. 
\end{proof}

\begin{cor}
    Assume $I\in \O(V)[[\hbar]]$ satisfies QME
    $$
    (Q+\hbar \Delta_0)e^{I/\hbar}=0
    $$
    in the DGBV $(\O(V),Q,\Delta_0)$. Then $\tilde{I}\in \O(V)[[\hbar]]$  satisfies QME
     $$
    (Q+\hbar \Delta_P)e^{\tilde{I}/\hbar}=0
    $$
    in the  DGBV $(\O(V),Q,\Delta_P)$. Here $\tilde{I}$ is related to I by
    $$
    e^{\tilde{I}/\hbar}=e^{\hbar \p_P}e^{I/\hbar}
    $$   
\end{cor}

 The operator $e^{\hbar \p_P}$ plays the role of integration with respect to the Gaussian measure. The relation $e^{\tilde{I}/\hbar}=e^{\hbar \p_P}e^{I/\hbar}$ can be read via Wick's Theorem as 
 
\bea 
\tikzset{every picture/.style={line width=0.75pt}}   
\begin{tikzpicture}[x=0.75pt,y=0.75pt,yscale=-1,xscale=1]

\draw    (192.5,32) -- (224.72,64.22) ;
\draw [shift={(224.72,64.22)}, rotate = 45] [color={rgb, 255:red, 0; green, 0; blue, 0 }  ][fill={rgb, 255:red, 0; green, 0; blue, 0 }  ][line width=0.75]      (0, 0) circle [x radius= 3.35, y radius= 3.35]   ;
\draw    (190.61,98.34) -- (224.72,64.22) ;
\draw [shift={(224.72,64.22)}, rotate = 315] [color={rgb, 255:red, 0; green, 0; blue, 0 }  ][fill={rgb, 255:red, 0; green, 0; blue, 0 }  ][line width=0.75]      (0, 0) circle [x radius= 3.35, y radius= 3.35]   ;
\draw    (318.27,96.34) -- (285.97,64.2) ;
\draw [shift={(285.97,64.2)}, rotate = 224.86] [color={rgb, 255:red, 0; green, 0; blue, 0 }  ][fill={rgb, 255:red, 0; green, 0; blue, 0 }  ][line width=0.75]      (0, 0) circle [x radius= 3.35, y radius= 3.35]   ;
\draw    (320,30) -- (285.97,64.2) ;
\draw [shift={(285.97,64.2)}, rotate = 134.86] [color={rgb, 255:red, 0; green, 0; blue, 0 }  ][fill={rgb, 255:red, 0; green, 0; blue, 0 }  ][line width=0.75]      (0, 0) circle [x radius= 3.35, y radius= 3.35]   ;
\draw   (224.72,64.13) .. controls (224.72,47.07) and (238.55,33.25) .. (255.6,33.25) .. controls (272.65,33.25) and (286.48,47.07) .. (286.48,64.13) .. controls (286.48,81.18) and (272.65,95) .. (255.6,95) .. controls (238.55,95) and (224.72,81.18) .. (224.72,64.13) -- cycle ;

\draw (174.92,67) node  [font=\fontsize{4.71em}{5.65em}\selectfont]  {$($};
\draw (338.42,66) node  [font=\fontsize{4.71em}{5.65em}\selectfont]  {$)$};
\draw (38,49.4) node [anchor=north west][inner sep=0.75pt]    {$\tilde{I} =\sum\limits_{\text{connected graphs}}$};
\draw (251,12.4) node [anchor=north west][inner sep=0.75pt]    {$P$};
\draw (247,102.4) node [anchor=north west][inner sep=0.75pt]    {$P$};
\draw (206,57.4) node [anchor=north west][inner sep=0.75pt]    {$I$};
\draw (296,57.4) node [anchor=north west][inner sep=0.75pt]    {$I$};
\draw (358,62.4) node [anchor=north west][inner sep=0.75pt]    {$.$};
\end{tikzpicture}
\eea
Here $I$ serves as vertices and $P$ for the propagator. Thus, Feynman diagrams give the required chain homotopy between different DGBV's.

\subsubsection*{Back to QFT} 
Now consider the QFT set-up $(\EE=\Gamma(X,E^\blt),Q,\omega)$. The problem is that the Poisson kernel $K_0=\omega\inv$ is a $\delta$-function distribution which leads to a singular BV operator. Nevertheless we know $K_0$ is $Q$-closed
\begin{align*}
Q(K_0)=0. 
\end{align*}
\ut{Costello's approach}: Using elliptic regularity
$$
H^\blt(\text{Distribution},Q)=H^\blt(\text{Smooth},Q). 
$$
Here ``Distribution" or ``Smooth" means distributional or smooth sections of relevant tensor bundles of $E^\blt$. Thus we can replace $K_0$ by a smooth object in its $Q$-cohomology class by
$$
K_0=K_r+Q(P_r). 
$$
Here $K_r$ is smooth while $P_r$ (called parametrix) is singular. Define
\begin{align*}
    &\Delta_r:\text{BV operator associated with $K_r$}
\end{align*}
Since $K_r\in \Sym^2(\E)$ is now smooth, the operator
$$
\Delta_r: \O(\EE)\to \O(\EE)\text{ is well-defined}.
$$
\begin{defn}
The DGBV $(\O(\EE),Q,\Delta_r)$ will be called the \textbf{effective DGBV} with respect to the regularization $r$. 
\end{defn}

Let $r'$ be another regularization with parametrix $P_{r'}$
$$
K_0=K_{r'}+Q(P_{r'}). 
$$
Then the two regularized Poisson kernels differ by a chain homotopy
$$
K_{r'}-K_{r}=Q(P^{r'}_r)
$$
where $P^{r'}_r\in \Sym^2(\EE)$ is smooth. Let
$$
\p_{P^{r'}_r}:\O(\EE)\lra\O(\EE)
$$
be the 2nd order operator of contracting with the smooth kernel $P^{r'}_r$.

The same argument as in the toy model gives the chain homotopy

\[\begin{tikzcd}
	{\big(\O(\EE)[[\hbar]],Q+\hbar\Delta_r\big)} & {} & {\big(\O(\EE)[[\hbar]],Q+\hbar\Delta_r^\prime}\big) \\
	& {{\text{Homotopy RG flow (HRG)}}}
	\arrow["\exp\bracket{\hbar\partial_{P^{r^\prime}_r}}", from=1-1, to=1-3]
	\arrow[from=2-2, to=1-2]
\end{tikzcd}\]

\begin{defn}[Costello\cite{costello2011renormalization}]
    An effective perturbative quantization of $I_0$ (which satisfies CME) is a family 
    $$
    I[r]\in \O(\EE)[[\hbar]]\quad 
    $$
    (whichi is at least cubic modulo $\hbar$) for each choice of regularization $r$ satisfying
    \begin{itemize}
        \item Effective QME
        $$
        (Q+\hbar\Delta_r)e^{I[r]/\hbar}=0.
        $$ 
        \item Homotopy RG flow $$
e^{I[r']/\hbar}=e^{\hbar\p_{P^{r'}_r}}e^{I[r]/\hbar}
$$ 
which is equivalent to the Feynman diagram expansion

\bea 
\tikzset{every picture/.style={line width=0.75pt}}   
\begin{tikzpicture}[x=0.75pt,y=0.75pt,yscale=-1,xscale=1]

\draw    (192.5,32) -- (224.72,64.22) ;
\draw [shift={(224.72,64.22)}, rotate = 45] [color={rgb, 255:red, 0; green, 0; blue, 0 }  ][fill={rgb, 255:red, 0; green, 0; blue, 0 }  ][line width=0.75]      (0, 0) circle [x radius= 3.35, y radius= 3.35]   ;
\draw    (190.61,98.34) -- (224.72,64.22) ;
\draw [shift={(224.72,64.22)}, rotate = 315] [color={rgb, 255:red, 0; green, 0; blue, 0 }  ][fill={rgb, 255:red, 0; green, 0; blue, 0 }  ][line width=0.75]      (0, 0) circle [x radius= 3.35, y radius= 3.35]   ;
\draw    (318.27,96.34) -- (285.97,64.2) ;
\draw [shift={(285.97,64.2)}, rotate = 224.86] [color={rgb, 255:red, 0; green, 0; blue, 0 }  ][fill={rgb, 255:red, 0; green, 0; blue, 0 }  ][line width=0.75]      (0, 0) circle [x radius= 3.35, y radius= 3.35]   ;
\draw    (320,30) -- (285.97,64.2) ;
\draw [shift={(285.97,64.2)}, rotate = 134.86] [color={rgb, 255:red, 0; green, 0; blue, 0 }  ][fill={rgb, 255:red, 0; green, 0; blue, 0 }  ][line width=0.75]      (0, 0) circle [x radius= 3.35, y radius= 3.35]   ;
\draw   (224.72,64.13) .. controls (224.72,47.07) and (238.55,33.25) .. (255.6,33.25) .. controls (272.65,33.25) and (286.48,47.07) .. (286.48,64.13) .. controls (286.48,81.18) and (272.65,95) .. (255.6,95) .. controls (238.55,95) and (224.72,81.18) .. (224.72,64.13) -- cycle ;

\draw (174.92,67) node  [font=\fontsize{4.71em}{5.65em}\selectfont]  {$($};
\draw (338.42,66) node  [font=\fontsize{4.71em}{5.65em}\selectfont]  {$)$};
\draw (38,49.4) node [anchor=north west][inner sep=0.75pt]    {$I[r'] =\sum\limits_{\text{connected graphs}}$};
\draw (251,12.4) node [anchor=north west][inner sep=0.75pt]    {$P_r^{r'}$};
\draw (247,102.4) node [anchor=north west][inner sep=0.75pt]    {$P_r^{r'}$};
\draw (196,57.4) node [anchor=north west][inner sep=0.75pt]    {$I[r]$};
\draw (296,57.4) node [anchor=north west][inner sep=0.75pt]    {$I[r]$};
\draw (358,62.4) node [anchor=north west][inner sep=0.75pt]    {$.$};
\end{tikzpicture}
\eea
       
    \item  $I[r]$ is asymptotic local when $r\lra 0$ and has the classical limit
    \[
    \lim_{r\rightarrow0}I_0[r]=I_0
    \]
    \end{itemize}
\end{defn}

Here is a pictue to illustrate what is going on. The situation is very similar to how residue is defined in algebraic geometry: we need to perturb the singularity and define residue at the deformed configuration, and show that all local deformations give the same answers. Here we use all "nearby" regularizations to define the unrenormalized point. 

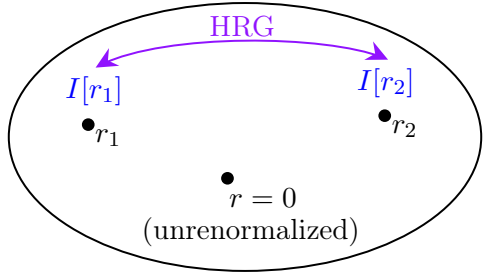
\bea 
\tikzset{every picture/.style={line width=0.75pt}}         
\begin{tikzpicture}[x=0.75pt,y=0.75pt,yscale=-1,xscale=1]

\draw [color={rgb, 255:red, 144; green, 19; blue, 254 }  ,draw opacity=1 ]   (230.27,95.93) .. controls (262.54,80.07) and (341.87,80.98) .. (370.8,92.41) ;
\draw [shift={(373.33,93.5)}, rotate = 205.15] [fill={rgb, 255:red, 144; green, 19; blue, 254 }  ,fill opacity=1 ][line width=0.08]  [draw opacity=0] (10.72,-5.15) -- (0,0) -- (10.72,5.15) -- (7.12,0) -- cycle    ;
\draw [shift={(227.33,97.5)}, rotate = 329.86] [fill={rgb, 255:red, 144; green, 19; blue, 254 }  ,fill opacity=1 ][line width=0.08]  [draw opacity=0] (10.72,-5.15) -- (0,0) -- (10.72,5.15) -- (7.12,0) -- cycle    ;
\draw  [fill={rgb, 255:red, 0; green, 0; blue, 0 }  ,fill opacity=1 ] (221,126.17) .. controls (221,124.69) and (222.19,123.5) .. (223.67,123.5) .. controls (225.14,123.5) and (226.33,124.69) .. (226.33,126.17) .. controls (226.33,127.64) and (225.14,128.83) .. (223.67,128.83) .. controls (222.19,128.83) and (221,127.64) .. (221,126.17) -- cycle ;
\draw  [fill={rgb, 255:red, 0; green, 0; blue, 0 }  ,fill opacity=1 ] (369,121.67) .. controls (369,120.19) and (370.19,119) .. (371.67,119) .. controls (373.14,119) and (374.33,120.19) .. (374.33,121.67) .. controls (374.33,123.14) and (373.14,124.33) .. (371.67,124.33) .. controls (370.19,124.33) and (369,123.14) .. (369,121.67) -- cycle ;
\draw  [fill={rgb, 255:red, 0; green, 0; blue, 0 }  ,fill opacity=1 ] (290.5,153.17) .. controls (290.5,151.69) and (291.69,150.5) .. (293.17,150.5) .. controls (294.64,150.5) and (295.83,151.69) .. (295.83,153.17) .. controls (295.83,154.64) and (294.64,155.83) .. (293.17,155.83) .. controls (291.69,155.83) and (290.5,154.64) .. (290.5,153.17) -- cycle ;
\draw   (184,132.38) .. controls (184,95.17) and (236.79,65) .. (301.92,65) .. controls (367.04,65) and (419.83,95.17) .. (419.83,132.38) .. controls (419.83,169.59) and (367.04,199.76) .. (301.92,199.76) .. controls (236.79,199.76) and (184,169.59) .. (184,132.38) -- cycle ;

\draw (225.67,126.9) node [anchor=north west][inner sep=0.75pt]    {$r_{1}$};
\draw (373.67,122.4) node [anchor=north west][inner sep=0.75pt]    {$r_{2}$};
\draw (292.5,156.57) node [anchor=north west][inner sep=0.75pt]    {$r=0$};
\draw (210.87,100.18) node [anchor=north west][inner sep=0.75pt]  [color={rgb, 255:red, 0; green, 0; blue, 255 }  ,opacity=1 ]  {$I[ r_{1}]$};
\draw (356.42,96.3) node [anchor=north west][inner sep=0.75pt]  [color={rgb, 255:red, 0; green, 0; blue, 255 }  ,opacity=1 ]  {$I[ r_{2}]$};
\draw (282.77,69.83) node [anchor=north west][inner sep=0.75pt]  [color={rgb, 255:red, 144; green, 19; blue, 254 }  ,opacity=1 ] [align=left] {HRG};
\draw (249,171.5) node [anchor=north west][inner sep=0.75pt]   [align=left] {(unrenormalized)};
\end{tikzpicture}
\eea

In practice, here are steps for constructing perturbative quantization. 
\begin{itemize}
    \item [\tc{1}] Construct counter-term $I^\e\in\hbar\O_{loc}(\EE)[[\hbar]]$ such that 
    \[
e^{I[r]^{Naive}/\hbar}:=\lim_{\e\rightarrow0}(e^{\hbar{P^{r}_\e}}e^{(I_0+I^\e)/\hbar})\quad  \text{exists}
    \]
    Then this naive family
   $\lr{I[r]^{Naive}}_r$ satisfies HRG by construction.
   \item[\tc{2}] The choice of counter-terms is not unique. We need to further correct $I^\e $ such that $e^{I[r]/\hbar}$ satisfies QME.
\end{itemize}
\tc{1} is always possible by the method of counter-term. \tc{2} is \ut{NOT} always possible: obstruction may exist which is called "\textbf{gauge anomaly}" in physics terminology. There is a deformation-obstruction theory, which shows that the gauge anomaly lies in 
\[
H^1(\O_{loc}(\EE),Q+\{I_0,-\}).
\]

\subsection{Heat Kernel Regularization}
There are many ways of regularizations. One  method that connects to geometry is the heat kernel regularization. Typically, fixing a choice of metric, we have 
\begin{itemize}
    \item the adjoint of the elliptic operator $Q:\cE\to \cE$, denoted as $Q^\dagger: \cE\to \cE$,
    \item assume $\lsb Q,Q^\dagger\rsb= QQ^\dagger+Q^\dagger Q$ is a generalized Laplacian \footnote{ We use $[-,-]$ for graded commutator in this paper. }.
\end{itemize}

Thus we can define a heat operator $e^{-L\lsb Q,Q^\dagger\rsb}$ for $L>0$. Let $K_L\in \sym^2(\cE)$ be the kernel of the heat operator by 
\bea \lb e^{-L\lsb Q,Q^\dagger\rsb}\alpha\rb \lb x\rb
=\int dy \lan K_L(x,y),\alpha(y)\ran\quad \text{for }\alpha\in\cE.\eea 
Here $\lan-,-\ran$ is the pairing from $\omega$.
Note that
\begin{itemize}
    \item $K_0=\lim\limits_{L\to 0} K_L$ is the $\delta$-function distribution $\omega^{-1}$,
    \item $K_L\in \sym^2(\cE)$ is smooth for $L>0$.
\end{itemize}

Let $P_L$ be the kernel of the operator $\int_0^L dt\ Q^\dagger e^{-t\lsb Q,Q^\dagger\rsb} $. Explicitly, we have 
\bea P_L= \int_0^L dt \lb Q^\dagger \otimes 1\rb K_t.\eea
The operator equation 
\bea \lsb Q, \int_0^L dt\ Q^\dagger e^{-t\lsb Q,Q^\dagger\rsb} \rsb 
= \int_0^L dt \lsb Q,Q^\dagger\rsb e^{-t\lsb Q,Q^\dagger\rsb} 
=1-e^{-L\lsb Q,Q^\dagger\rsb}\eea
can be translated into the kernel equation:
\bea K_0-K_L=(Q\otimes 1+1\otimes Q)P_L\eea
or simply written as
\bea K_0-K_L=Q \lb P_L\rb.\eea
We can use $K_L$ to define the effective QME. 

Similarly, for $0<\varepsilon<L$,  the operator equation is 
\bea \lsb Q, \int_\varepsilon^L dt\ Q^\dagger e^{-t\lsb Q,Q^\dagger\rsb} \rsb 
=e^{-\varepsilon\lsb Q,Q^\dagger\rsb}-e^{-L\lsb Q,Q^\dagger\rsb}\eea
or 
\bea K_\varepsilon-K_L=(Q\otimes 1+1\otimes Q)P_\varepsilon^L,\eea
where $P^L_\varepsilon= \int_\varepsilon^L dt \lb Q^\dagger \otimes 1\rb K_t$ is called the \emph{regularized propagator}. 
Now we can use $P^L_\varepsilon$ to connect the effective QME at $\varepsilon$ with the effective QME at $L$ via the HRG.

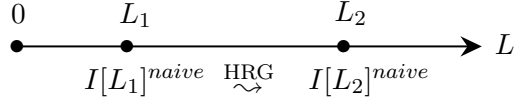
\bea
\tikzset{every picture/.style={line width=0.75pt}} 
\begin{tikzpicture}[x=0.75pt,y=0.75pt,yscale=-1,xscale=1]

\draw    (10.5,100.67) -- (237.5,100.67) ;
\draw [shift={(240.5,100.67)}, rotate = 180] [fill={rgb, 255:red, 0; green, 0; blue, 0 }  ][line width=0.08]  [draw opacity=0] (10.72,-5.15) -- (0,0) -- (10.72,5.15) -- (7.12,0) -- cycle    ;

\draw (4,78.23) node [anchor=north west][inner sep=0.75pt]    {$0$};
\draw (3.12,96) node [anchor=north west][inner sep=0.75pt]  [font=\normalsize,rotate=-0.88]  {$\blt$};
\draw (58,96) node [anchor=north west][inner sep=0.75pt]  [font=\normalsize,rotate=-0.88]  {$\blt$};
\draw (166.12,96) node [anchor=north west][inner sep=0.75pt]  [font=\normalsize,rotate=-0.88]  {$\blt$};
\draw (58,77.23) node [anchor=north west][inner sep=0.75pt]    {$L_{1}$};
\draw (165.5,76.23) node [anchor=north west][inner sep=0.75pt]    {$L_{2}$};
\draw (245.5,93.07) node [anchor=north west][inner sep=0.75pt]    {$L$};
\draw (40.5,108.07) node [anchor=north west][inner sep=0.75pt]    {$I[L_{1}]^{naive}$};
\draw (152.5,108.07) node [anchor=north west][inner sep=0.75pt]    {$I[L_{2}]^{naive}$};
\draw (110,106.57) node [anchor=north west][inner sep=0.75pt]    {$\overset{\text{HRG}}{\leadsto}$};
\end{tikzpicture}
\eea

\begin{rmk}
$P_0^\infty= \int_0^\infty dt \lb Q^\dagger \otimes 1\rb K_t$ is the \emph{full propagator}. At $t=0$, one will encounter ultraviolet (UV) divergence since there exists a singularity for the full propagator. On a non-compact manifold, one will encounter infrared (IR) divergence at $t=\infty$.
\end{rmk}

Consider the case when $X$ is compact. Let 
$$
\mathbf{H}=\lcb \left. \varphi\in\cE\ \right|\ \lsb Q,Q^\dagger\rsb\varphi=0\rcb\ = 
\lcb \left. \varphi\in\cE\ \right|\  Q\varphi=Q^\dagger\varphi=0\rcb\ \simeq H^\blt\lb \cE,Q\rb.
$$
$\mathbf{H}$ is called the space of {harmonics} (or the {zero modes}), which is a finite-dimensional space (by Hodge theory).
Then we have
\bea
\begin{tikzcd}
\infty-\text{dimensional } (-1)-\text{symplectic geometry } (\cE,Q,\omega)  \ar[d, "L\to\infty"] \\
\text{finite-dimensional } (-1)-\text{symplectic geometry } (\mathbf{H},\omega_{\mathbf{H}}
=\left. \omega \right|_{\mathbf{H}})
\end{tikzcd}
\eea

The BV operator $\Delta_{\mathbf{H}}$ associated with $\omega_{\mathbf{H}}^{-1}$ is $\Delta_{\mathbf{H}}=\Delta_\infty$. The essential story of effective BV quantization is depicted in the following diagram,
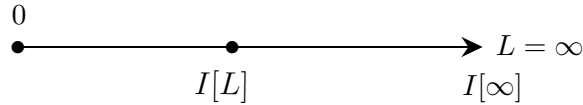
\bea
\tikzset{every picture/.style={line width=0.75pt}} 
\begin{tikzpicture}[x=0.75pt,y=0.75pt,yscale=-1,xscale=1]

\draw    (10.5,100.67) -- (237.5,100.67) ;
\draw [shift={(240.5,100.67)}, rotate = 180] [fill={rgb, 255:red, 0; green, 0; blue, 0 }  ][line width=0.08]  [draw opacity=0] (10.72,-5.15) -- (0,0) -- (10.72,5.15) -- (7.12,0) -- cycle    ;

\draw (4,78.23) node [anchor=north west][inner sep=0.75pt]    {$0$};
\draw (3.12,96) node [anchor=north west][inner sep=0.75pt]  [font=\normalsize,rotate=-0.88]  {$\blt$};
\draw (110.12,96) node [anchor=north west][inner sep=0.75pt]  [font=\normalsize,rotate=-0.88]  {$\blt$};
\draw (110.12,130.07) node [anchor=south][inner sep=0.75pt]  [font=\normalsize,rotate=-0.88]  {$I[L]$};

\draw (245.5,93.07) node [anchor=north west][inner sep=0.75pt]    {$L=\infty$};
\draw (245.5,130.07) node [anchor=south][inner sep=0.75pt]    {$I[\infty]$};

\end{tikzpicture}
\eea
and $I[\infty]$ solves the QME for $\lb \sO(\mathbf{H}),\Delta_{\mathbf{H}}\rb$ at $L=\infty$. The limit $L\to \infty$ is an interesting point where we will find some finite-dimensional geometric data.

\subsection{UV Finite Theory}\label{sec:UV}

In the BV formalism, the classical master equation
$$
QI_0+\frac{1}{2}\{I_0,I_0\}=0
$$
is quantized to the quantum master equation
$$
\text{``}QI+\hbar\Delta I+\frac{1}{2}\{I,I\}=0".
$$

As we explained as above, this naive quantum master equation is ill-defined for local $I\in \O_{loc}(\EE)$, and we have to use regularization to formulate the renormalized quantum master equation
$$
QI[r]+\hbar\Delta_r I[r]+\frac{1}{2}\{I,I\}_r=0.
$$

If the effective action at regularization $r$ can be found as
$$
e^{I[r]/\hbar}=\lim_{\e\rightarrow 0}e^{\hbar\p_{P^r_\e}}e^{I/\hbar}
$$
for $I\in \O_{loc}(\e)[[\hbar]]$, i.e., the $\e$-dependent counter-term is \textbf{\underline{NOT}} needed, we say the theory is UV finite. That is, for all regularized Feynman diagrams
\begin{center}
    \begin{tikzpicture}[scale=0.7]
        \coordinate (x) at (0,0);
        \coordinate (y) at (3,0);
        \coordinate (x1) at (-1,1);
        \coordinate (x2) at (-1,-1);
        \coordinate (y1) at (4,1);
        \coordinate (y2) at (4,-1);
        \fill (x) circle (2pt);
        \fill (y) circle (2pt);

        \draw[thick] (x) to[out=45,in=135] (y);
        \draw[thick] (x) to[out=-45,in=-135] (y);
        \draw[thick] (x) to (x1);
        \draw[thick] (x) to (x2);
        \draw[thick] (y) to (y1);
        \draw[thick] (y) to (y2);

        \node at (-2,0) {$\displaystyle\lim_{\e\rightarrow 0}$};
         \node at (6,0) {exist.};
        \node[anchor=south] at (1.5,0.5) {$P_\e^L$};
        \node[anchor=north] at (0,-0.2) {$I$};
        \node[anchor=north] at (3,-0.2) {$I$};
    \end{tikzpicture}
\end{center}

In this way, we can consider the limit 
$$
I[r]\rightarrow I,\qquad r\rightarrow 0,
$$
and the $r\ra 0$ limit of the renormalized quantum master equation 
$$
QI+\frac{1}{2}\{I,I\}+\cdots=0
$$
will have a local expression that deforms the CME. 
\begin{conj}
For UV finite theory, we expect to describe effective QME at $r\to 0$ limit by
$$
l_1^\hbar I+\frac{1}{2}l_2^\hbar(I,I)+\frac{1}{3!}l_3^\hbar (I,I,I)+\cdots=0
$$
where $\{l_1^\hbar,l_2^\hbar,\cdots\}$ defines a family of $L_\infty$-algebra parametrized by $\hbar$. They can be viewed as traded from $\Delta$ in terms of the renormalization procedure.
\end{conj}

There are two main classes of UV finite theories.
\begin{itemize}
    \item [\tc 1] Topological theory (Chern-Simons type)  where $\EE$ is of the form of de Rham complex. The UV finite property was established by Kontsevich \cite{Kont-diagram} and Axelrod-Singer \cite{AS-CS} using the compactified configuration space.
    \item [\tc 2] Holomorphic theory where $\EE$ is of the form of Dolbeault complex.  In this case, the Feynman graph integral can not be extended to the compactified configuration space. Fortunately, the UV finite property still holds in general. 
    \begin{itemize}
        \item $\dim_{\C}=1$: the UV property for chiral deformations is known to physicists via the method of point-splitting regularization (see for example Douglas \cite{Douglas} and Dijkgraaf \cite{Dijkgraaf-Chiral}). This method is essentially Cauchy principal value, and a homological theory for such regularization was systematically developed in Li-Zhou \cite{LZ-1}. In the framework of effective BV quantization, the UV finite property was established in  Li \cite{LS-D,LS-VOA}.
        
        \item $\dim_{\C}>1$: the method in \cite{LS-VOA} has been generalized for one loop graphs in Costello-Li \cite{Open-closed} and Williams \cite{Williams}. At higher loops,  Budzik-Gaiotto-Kulp-Wu-Yu \cite{Operatope} presented a strategy to prove UV finiteness for Laman graphs. In \cite{Wang-Cn}, Wang proved the UV finite property for all graphs on all $\C^n$ using a compactified Schwinger space. This is further generalized to K\"{a}hler manifolds in Wang-Yan \cite{Wang-Yan}.
    \end{itemize}
\end{itemize}

It is an extremely interesting question to figure out  $\{l_1^\hbar,l_2^\hbar,\cdots\}$ in these examples. In Section \ref{sec:TQM} and Section \ref{sec:2d1}, we explain the simplest example in each of these two categories  (Conjecture holds there) to illustrate the underlying rich structures.

\section{Topological Quantum Mechanics}\label{sec:TQM}

In this section we consider the example of topological quantum mechanics and illustrate its connection with deformation quantization and algebraic index theorem. 

\subsection{Deformation Quantization}
The method of deformation quantization was developed in the series of papers by Bayen, Flato, Fronsdal, Lichnerowicz, Sternheimer \cite{BFFLS}. The space of the real-valued (or complex-valued) functions on a phase space admits
two algebraic structures: a structure of \emph{associative algebra} given by the usual product of
functions and a structure of Lie algebra given by the \emph{Poisson bracket}. The study of the
properties of the deformations (in a suitable sense) of these two structures gives a new
invariant approach for quantum mechanics.
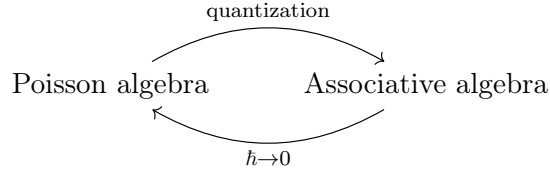
\bea
\begin{tikzcd}
\text{Poisson algebra}
\arrow[r, bend left, "\text{quantization}"] & \arrow[l, bend left, "\hbar\to 0"] \text{Associative algebra} 
\end{tikzcd}\eea
This is essentially the quantization method in quantum mechanics, in which a function $f$ on the classical phase space is quantized to an operator $\widehat{f}$.

\begin{defn}
A \textbf{Poisson manifold} is a pair $\lb X,P\rb$, where $X$ is a smooth manifold, and $P\in \Gamma(X, \asym^2 TX)$ satisfying $\lcb P,P\rcb_{\text{SN}}=0$.
\end{defn}

Here $\lcb -,-\rcb_{\text{SN}}$ is the \emph{Schouten-Nijenhuis bracket}. $P$ is called the {Poisson tensor/bi-vector}. In local coordinates, we can write 
\bea P=\sum_{i,j} P^{ij}(x)\p_i \wedge \p_j.\eea
It defines a \emph{Poisson bracket} $\lcb-,-\rcb_P$ on $C^\infty(X)$:
\bea \lcb f,g\rcb_P \coloneqq \sum_{i,j} P^{ij}\p_i f \p_j g,\quad \forall f,g\in C^\infty(X).\eea
The Poisson condition $\lcb P,P\rcb_{\text{SN}}=0$ implies $\lcb-,-\rcb_P$ satisfies Jacobi identity. Hence $\lcb-,-\rcb_P$ naturally defines the Poisson algebra $\lb C^\infty(X), \lcb-,-\rcb_P\rb$.

\begin{eg}
Let $(X,\omega)$ be a symplectic manifold, where $\omega= \hf\sum_{i,j} \omega_{ij} dx^i \wedge dx^j$ is the symplectic 2-form. Let 
$$
P=\omega^{-1}=\hf \sum_{i,j} \omega^{ij} \p_i \wedge \p_j,
$$
where $(\omega^{ij})$ is the inverse of $(\omega_{ij})$. Then 
$$
d\omega=0 \LRA \lcb P,P\rcb_{\text{SN}}=0.
$$
Hence $(X,\omega^{-1})$ defines a Poisson manifold.
\end{eg}

\begin{defn}
A \textbf{star-product} on a Poisson manifold $(X,P)$ is a $\bR[[ \hbar]]$-bilinear map
\begin{align*} 
C^\infty(X)[[ \hbar]] \times C^\infty(X)[[ \hbar]] &\to C^\infty(x)[[ \hbar]]\\
f \times g \quad &\mapsto f\star g=\sum_{k\geq 0} \hbar^k c_k\lb f,g\rb 
\end{align*}
such that
\bi[(1)]
\item $\star$ is associative: $\lb f\star g\rb \star h=f\star \lb g\star h\rb$,
\item $f\star g= fg+\cO(\hbar), \quad \forall f,g\in C^\infty(X)$,
\item $\hf \lb f\star g-g\star f\rb=\hbar\lcb f,g\rcb+\cO(\hbar^2), \quad \forall f,g\in C^\infty(X)$,
\item $c_k: C^\infty(X)\times C^\infty(X)\to C^\infty(X)$ is a bidifferential operator. 
\ei
Then $\lb C^\infty(X)[[\hbar]], \star\rb$ is called a \textbf{deformation quantization} of $\lb X,P\rb$.
\end{defn}

The definition of deformation quantization is purely algebraic. The \emph{existence} of deformation quantization is highly nontrivial. DeWilde-Lecomte first \cite{DeWilde-Lecomte} obtained the general existence of deformation quantization on symplectic manifolds via cohomological method.  Fedosov \cite{Fedosov-DQ} presented another beautiful approach on symplectic manifolds via differential geometric method. In the general case, Kontsevich \cite{Kont-DQ} gave the complete solution for arbitrary Poisson manifold. The parameter $\hbar$ is formal in the above definition of deformation quantization (only formal power series of $\hbar$ is concerned). There is also a notion of strict deformation quantization introduced by Rieffel \cite{Rieffel} in terms of $C^*$-algebras where $\hbar$ is not formal. 

\begin{eg} Let $X=\bR^{2n}$, with Poisson tensor  
$$
P=\hf \sum_{i,j} P^{ij} \p_i\wedge \p_j
$$ 
Here $P^{ij}$ are constants. Given $f(x),g(x)$, define the (formal) Moyal product $\star$ by
\bea
(f\star g)(x)=\left. \exp{\lb \frac{\hbar}{2} \sum_{i,j} P^{ij} \frac{\p}{\p y^i}\frac{\p}{\p z^j}\rb}\right|_{y=z=x} f(y) g(z)
\eea
or pictorially,
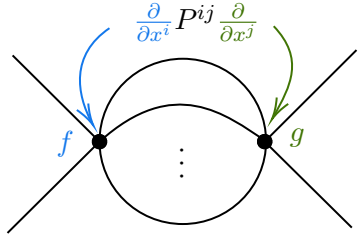
\bea 
\tikzset{every picture/.style={line width=0.75pt}}   
\begin{tikzpicture}[x=0.75pt,y=0.75pt,yscale=-1,xscale=1]

\draw    (123.16,67.69) -- (166.54,111.07) ;
\draw [shift={(166.54,111.07)}, rotate = 45] [color={rgb, 255:red, 0; green, 0; blue, 0 }  ][fill={rgb, 255:red, 0; green, 0; blue, 0 }  ][line width=0.75]      (0, 0) circle [x radius= 3.35, y radius= 3.35]   ;
\draw    (120.61,157) -- (166.54,111.07) ;
\draw [shift={(166.54,111.07)}, rotate = 315] [color={rgb, 255:red, 0; green, 0; blue, 0 }  ][fill={rgb, 255:red, 0; green, 0; blue, 0 }  ][line width=0.75]      (0, 0) circle [x radius= 3.35, y radius= 3.35]   ;
\draw    (292.48,154.31) -- (248.99,111.04) ;
\draw [shift={(248.99,111.04)}, rotate = 224.86] [color={rgb, 255:red, 0; green, 0; blue, 0 }  ][fill={rgb, 255:red, 0; green, 0; blue, 0 }  ][line width=0.75]      (0, 0) circle [x radius= 3.35, y radius= 3.35]   ;
\draw    (294.81,65) -- (248.99,111.04) ;
\draw [shift={(248.99,111.04)}, rotate = 134.86] [color={rgb, 255:red, 0; green, 0; blue, 0 }  ][fill={rgb, 255:red, 0; green, 0; blue, 0 }  ][line width=0.75]      (0, 0) circle [x radius= 3.35, y radius= 3.35]   ;
\draw   (166.54,110.94) .. controls (166.54,87.98) and (185.15,69.37) .. (208.11,69.37) .. controls (231.06,69.37) and (249.68,87.98) .. (249.68,110.94) .. controls (249.68,133.9) and (231.06,152.51) .. (208.11,152.51) .. controls (185.15,152.51) and (166.54,133.9) .. (166.54,110.94) -- cycle ;
\draw    (166.54,111.07) .. controls (203.93,77.12) and (229.51,97.31) .. (249.68,110.94) ;
\draw [color={rgb, 255:red, 20; green, 120; blue, 235 }  ,draw opacity=1 ]   (171.5,54) .. controls (151.82,67.44) and (155.18,87.33) .. (162.56,101.29) ;
\draw [shift={(163.5,103)}, rotate = 240.26] [color={rgb, 255:red, 20; green, 120; blue, 235 }  ,draw opacity=1 ][line width=0.75]    (10.93,-3.29) .. controls (6.95,-1.4) and (3.31,-0.3) .. (0,0) .. controls (3.31,0.3) and (6.95,1.4) .. (10.93,3.29)   ;
\draw [color={rgb, 255:red, 65; green, 117; blue, 5 }  ,draw opacity=1 ]   (253.5,53) .. controls (268.94,69.41) and (262.03,83.02) .. (253.44,98.33) ;
\draw [shift={(252.5,100)}, rotate = 299.36] [color={rgb, 255:red, 65; green, 117; blue, 5 }  ,draw opacity=1 ][line width=0.75]    (10.93,-3.29) .. controls (6.95,-1.4) and (3.31,-0.3) .. (0,0) .. controls (3.31,0.3) and (6.95,1.4) .. (10.93,3.29)   ;

\draw (143.5,104.17) node [anchor=north west][inner sep=0.75pt]  [color={rgb, 255:red, 20; green, 120; blue, 235 }  ,opacity=1 ]  {$f$};
\draw (260.09,103.17) node [anchor=north west][inner sep=0.75pt]  [color={rgb, 255:red, 65; green, 117; blue, 5 }  ,opacity=1 ]  {$g$};
\draw (182,40) node [anchor=north west][inner sep=0.75pt]    {$\textcolor[rgb]{0.08,0.47,0.92}{\frac{\partial }{\partial x^{i}}} P^{ij}\textcolor[rgb]{0.25,0.46,0.02}{\frac{\partial }{\partial x^{j}}}$};
\draw (212,112) node [anchor=north west][inner sep=0.75pt]  [rotate=-90]  {$\cdots $};
\end{tikzpicture}
\eea
Then $\star$ defines a deformation quantization.
\end{eg}

\begin{rmk}
If $P^{ij}\neq$ constant, then the above formula does not work. How to correct this to a star product is the celebrated Formality Theorem in \cite{Kont-DQ}. 
\end{rmk}

\begin{defn}
Let $\lb V,\omega\rb$ be a linear symplectic space, where $V\simeq \bR^{2n}$ and $\omega: \asym^2 V\to \bR$ is a symplectic pairing. $P=\omega^{-1}$ is the Poison tensor. Write 
\bea\widehat{\sO}(V)\coloneqq \widehat{\sym}(V^\vee)=\prod_{k\geq 0} \sym^k(V^\vee),\eea
where $V^\vee=\on{Hom}(V,\bR)$ is the linear dual of $V$. Then the Moyal product defines an associative algebra $\lb \widehat{\sO}(V)[[ \hbar]],\star\rb$, called the \textbf{(formal) Weyl algebra}.
\end{defn}

\subsection{Fedosov's Geometric Method}
We will focus on \emph{symplectic manifolds} in the rest of this section. Fedosov \cite{Fedosov-DQ} gave a simple and geometric construction of deformation quantization as follows. 
On a symplectic manifold $(X,\omega)$, the tangent plane $T_p X$ at each point $p \in X$ 
is a linear symplectic space. Quantum fluctuations deform the algebra of functions on $T_p X$ to the
associated \emph{Weyl algebra}. These pointwise Weyl algebras form a vector bundle --- the \emph{Weyl bundle} $\cW(X)$ on $X$.
\begin{defn}
Let $(X,\omega)$ be a symplectic manifold.
We define the \textbf{Weyl bundle} 
\bea \cW(X)\coloneqq \prod_{k\geq 0}\sym^k\lb T^\ast X\rb[[ \hbar]].\eea
So at each point $p\in X$, its fiber is
\bea \left. \cW(X)\right|_p= \widehat{\sO}\lb T_p X\rb[[\hbar]].\eea
Here $\widehat{\sO}$ refers to formal power series functions. 
\end{defn}
A local section of $\cW(X)$ is 
\bea \sigma(x,y)=\sum_{k,l\geq 0}\hbar^k a_{k,i_1\cdots i_l}(x) y^{i_1}\cdots y^{i_l},\eea
where $\lcb x^i\rcb$ are the base coordinates and $\lcb y^i\rcb$ are the fiber coordinates; $a_{k,i_1\cdots i_l}(x)$ are smooth functions.
Since $\lb T_p X, \left.\omega\right|_{T_p X}\rb$ is linear symplectic, we have a fiberwise Moyal product, still denoted by $\star$. Thus $\lb \cW(X),\star\rb$ defines the $\infty$-dimensional bundle of algebras.

Let $\nabla$ be a connection on $TX$ which is torsion-free and compatible with $\omega$ (i.e. $\nabla \omega=0$). Such connection is called a \textbf{symplectic connection} (which always exists and is not unique). $\nabla$ induces a connection on all tensors. In particular, it defines a connection on $\cW(X)$, still denoted by $\nabla$. Its curvature is
\bea \nabla^2 \sigma=\frac{1}{\hbar}\lsb R_\nabla, \sigma\rsb_\star, \quad \forall \sigma\in \Gamma(X,\cW(X))\eea
where 
$$
R_\nabla=\frac{1}{4}R_{ijkl} y^i y^j dx^k \wedge dx^l \in \Omega^2(X,\cW(X))
$$ 
is a 2-form valued in the Weyl bundle $\cW(X)$; $R_{ijkl}=\omega_{im}R^m{}_{jkl}$ is a curvature form contracted with the symplectic pairing.

Given a sequence of closed 2-forms $\lcb \omega_k\rcb_{k\geq 1}$ on $X$, Fedosov showed that there exists a unique (up to gauge) connection on $\cW(X)$ of the form 
$$
\nabla+\frac{1}{\hbar}\lsb \gamma,-\rsb_\star
$$
where $\gamma\in\Omega^1\lb X, \cW(X)\rb$ is a $\cW(X)$-valued 1-form, satisfying certain initial conditions and the equation 
\bea \nabla\gamma+\frac{1}{2\hbar}\lsb \gamma,\gamma\rsb_\star+R_\nabla=\omega_\hbar\qquad \text{(Fedosov equation)}\eea
where $\omega_{\hbar}= -\omega+ \sum_{k\geq 1}\hbar^k \omega_k$. 
Let 
\bea D= \nabla+\frac{1}{\hbar}\lsb \gamma,-\rsb_\star\eea
be the {Fedosov connection}. Then the Fedosov equation implies 
$$
D^2=\frac{1}{\hbar}\lsb\omega_\hbar,-\rsb_\star=0
$$
since $\omega_\hbar$ is constant along each fiber, thus a central term.
So we obtain a flat connection $D$ on $\cW(X)$. 
Fedosov equation has the geometric
interpretation of BV quantum master equation \cite{GLL-BV, GLX-Effective}.

Let $\cW_D(X)\coloneqq \lcb \left.\sigma\in\Gamma(X,\cW(X))\ \right|D\sigma=0 \rcb$ be the space of flat sections. Then $\lb \cW_D(X),\star\rb$ is an associative algebra.
Let 
\bea\sigma: \cW_D(X)\to C^\infty(X)[[\hbar]]\eea
be the {symbol map} by sending $y\mapsto 0$. Then $\sigma$ is an isomorphism, and 
\bea f\star g \mapsto \sigma\lb \sigma^{-1}\lb f\rb\star \sigma^{-1}\lb g\rb\rb\eea
defines a {deformation quantization}. $\omega_\hbar$ is the corresponding {characteristic class} (or {moduli}).

\subsection{Algebraic Index Theorem}
Given a deformation quantization $\lb C^\infty(X)[[\hbar]],\star \rb$ on a symplectic manifold with characteristic class $\omega_\hbar$, there exists a unique \textbf{trace map}
\bea \Tr: C^\infty(X)[[\hbar]] \to \bR ((\hbar))\eea
satisfying a normalization condition and the trace property:
\bea \Tr \lb f\star g\rb=\Tr \lb g\star f\rb.\eea

Then the index is obtained as the partition function of the theory, which can be formulated as 
\bea \on{Index}=\Tr(1)=\int_X e^{\omega_\hbar/\hbar} \widehat{A}(X),\eea
where $\widehat{A}(X)$ is the \emph{(formal) $\widehat{A}$-genus} of $X$.
This is the simplest version of \textbf{algebraic index theorem} formulated by Fedosov \cite{Fedosov-book} and Nest-Tsygan \cite{Nest-Tsygan} as the algebraic analogue of {Atiyah-Singer index theorem}. 

In the case of a vector bundle $E$ on $X$, one can similarly construct a \emph{deformation quantization} for $C^\infty\lb X,\on{End}(E)\rb[[\hbar]]$ and construct the trace map. In this case
$$
\Tr(1)=\int_X e^{\omega_\hbar/\hbar}\on{Ch}(E) \widehat{A}(X),
$$
where $\on{Ch}(E)$ is the Chern character of the vector bundle $E$ over $X$.

\subsubsection*{Relation with QFT}
In supersymmetric (SUSY) QFT, \emph{localization} often appears, in which the path integral on $\cE$ is often localized effectively to an equivalent integral on a finite-dimensional space $M\subset \cE$ describing some interesting moduli space:
\bea \int_\cE e^{iS/\hbar}= \int_M \lb -\rb.\eea

In topological QM, we find
\bea \int_{\on{Map}(S^1,X)} e^{-S/\hbar}\overset{\hbar\to 0}{=} \int_X \lb -\rb,\eea
where $\on{Map}(S^1,X)$ is a loop space, and $\int_X$ indicates the localization to a \emph{constant map}.
The path integral will be captured exactly by an effective theory in the formal neighborhood of
constant maps inside the full mapping space. Exact semi-classical approximation in $\hbar\to 0$ allows
us to reduce the path integral into a meaningful integral on the moduli space of constant maps, i.e., $X$.
The left-hand side usually gives a physics presentation of the analytic index of certain elliptic operator; the
right-hand side will end up with integrals of various curvature forms representing the topological index.

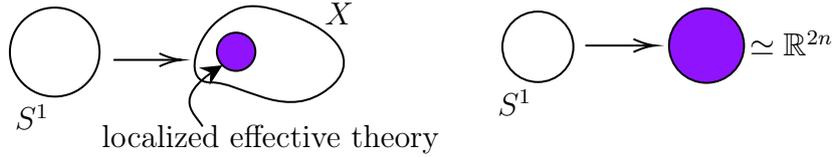
\begin{figure}[!htpb]\centering 
\tikzset{every picture/.style={line width=0.75pt}}         
\begin{tikzpicture}[x=0.75pt,y=0.75pt,yscale=-1,xscale=1]
\draw   (34.83,63.04) .. controls (34.83,50.68) and (44.85,40.67) .. (57.21,40.67) .. controls (69.57,40.67) and (79.59,50.68) .. (79.59,63.04) .. controls (79.59,75.4) and (69.57,85.42) .. (57.21,85.42) .. controls (44.85,85.42) and (34.83,75.4) .. (34.83,63.04) -- cycle ;
\draw    (86.58,66.97) -- (118.38,66.97) ;
\draw [shift={(120.38,66.97)}, rotate = 180] [color={rgb, 255:red, 0; green, 0; blue, 0 }  ][line width=0.75]    (10.93,-3.29) .. controls (6.95,-1.4) and (3.31,-0.3) .. (0,0) .. controls (3.31,0.3) and (6.95,1.4) .. (10.93,3.29)   ;
\draw   (138.44,47.45) .. controls (149.5,38.5) and (157,37) .. (179.23,45.99) .. controls (201.46,54.98) and (210.5,71.5) .. (190.3,83.87) .. controls (170.1,96.23) and (149.5,78) .. (135.24,78.33) .. controls (120.97,78.66) and (127.38,56.39) .. (138.44,47.45) -- cycle ;
\draw  [fill={rgb, 255:red, 144; green, 19; blue, 254 }  ,fill opacity=1 ] (138.15,62.89) .. controls (138.15,57.58) and (142.45,53.27) .. (147.76,53.27) .. controls (153.07,53.27) and (157.38,57.58) .. (157.38,62.89) .. controls (157.38,68.2) and (153.07,72.5) .. (147.76,72.5) .. controls (142.45,72.5) and (138.15,68.2) .. (138.15,62.89) -- cycle ;
\draw    (131,100.5) .. controls (127.97,94.38) and (114.35,86.96) .. (138.6,71.18) ;
\draw [shift={(140.96,69.69)}, rotate = 148.53] [fill={rgb, 255:red, 0; green, 0; blue, 0 }  ][line width=0.08]  [draw opacity=0] (10.72,-5.15) -- (0,0) -- (10.72,5.15) -- (7.12,0) -- cycle    ;
\draw   (280.67,60.39) .. controls (280.67,50.6) and (288.6,42.67) .. (298.39,42.67) .. controls (308.18,42.67) and (316.11,50.6) .. (316.11,60.39) .. controls (316.11,70.18) and (308.18,78.11) .. (298.39,78.11) .. controls (288.6,78.11) and (280.67,70.18) .. (280.67,60.39) -- cycle ;
\draw    (322.07,59.76) -- (355.03,59.76) ;
\draw [shift={(357.03,59.76)}, rotate = 180] [color={rgb, 255:red, 0; green, 0; blue, 0 }  ][line width=0.75]    (10.93,-3.29) .. controls (6.95,-1.4) and (3.31,-0.3) .. (0,0) .. controls (3.31,0.3) and (6.95,1.4) .. (10.93,3.29)   ;
\draw  [fill={rgb, 255:red, 144; green, 19; blue, 254 }  ,fill opacity=1 ] (363.86,59.81) .. controls (363.86,49.48) and (372.24,41.1) .. (382.57,41.1) .. controls (392.91,41.1) and (401.28,49.48) .. (401.28,59.81) .. controls (401.28,70.15) and (392.91,78.52) .. (382.57,78.52) .. controls (372.24,78.52) and (363.86,70.15) .. (363.86,59.81) -- cycle ;

\draw (35.94,87.75) node [anchor=north west][inner sep=0.75pt]    {$S^{1}$};
\draw (79.29,98.35) node [anchor=north west][inner sep=0.75pt]   [align=left] {localized effective theory};
\draw (190.64,36.72) node [anchor=north west][inner sep=0.75pt]    {$X$};
\draw (276.88,80.5) node [anchor=north west][inner sep=0.75pt]    {$S^{1}$};
\draw (402.6,50.7) node [anchor=north west][inner sep=0.75pt]    {$\simeq \mathbb{R}^{2n}$};
\end{tikzpicture}
\end{figure}

Geometrically, a loop space is mapped to a localized neighborhood of a point in $X$ (specified by the constant map), where a localized effective theory exists.
Locally, by the choice of Darboux coordinates, $X$ can be thought of as a standard phase space, $\bR^{2n}$. 
The loop spaces are then glued together on $X$ as a family of effective field theory. This can be done rigorously within the framework of effective BV quantization \cite{GLX-Effective}: we find the following dictionary 
\begin{itemize}
    \item Effective action $\leadsto \gamma$,
    \item Quantum master equation $\leadsto$ Fedosov equation,
    \item BV integral $\leadsto$ trace map,
    \item Partition function $\leadsto$ algebraic index.
\end{itemize}

\subsection{Local Theory}
\label{sec:tqm1}
In this section we study topological quantum mechanics in terms of the effective renormalization method, and explain how to use it to prove the algebraic index theorem. We follow the presentation in \cite{GLL-BV,GLX-Effective}.

\subsubsection*{Local model}
Let us consider the standard phase space $(V,\omega)$, where $V\simeq \bR^{2n}$ with coordinates 
$$
(x^1,\cdots,x^n,x^{n+1},\cdots, x^{2n})=(q^1, \cdots, q^n, p_1,\cdots,p_n)
$$ 
and $$
\omega=\sum_{i=1}^n dp_i\wedge dq^i.
$$

Let $S^1_{dR}$ be the (locally ringed) space with underlying topology of the circle $S^1$ and the structure sheaf $\sO(S^1_{dR})=\Omega^\blt_{S^1}$, which is a differential graded ring of differential forms with the de Rham differential operator $d$. Consider the \emph{local model} describing the space of maps
\bea \varphi: S^1_{dR}\to V\simeq \bR^{2n}.\eea
Such a $\varphi$ can be identified with an element in $\Omega^\blt_{S^1}\otimes V$. Explicitly, let $\theta$ be the coordinates on $S^1$ (with the identification $\theta\sim \theta+1$). The space of maps can then be written as 
\bea \lcb\varphi \rcb= \lcb \bP_i(\theta),\bQ^i(\theta)\rcb_{i=1,\cdots,n}, \quad \bP_i,\bQ^i\in \Omega^\blt_{S^1}.\eea
Writing in form component,
\bea \bP_i(\theta)=p_i(\theta)+\eta_i(\theta)d\theta, \quad \bQ^i(\theta)=q^i(\theta)+\xi^i(\theta)d\theta.\eea

So the space of fields is 
\bea \cE= \Omega^\blt_{S^1}\otimes V.\eea
The triple $\lb \Omega^\blt_{S^1}\otimes V, d, \int_{S^1} \lan -,-\ran_\omega\rb$ is an $\infty$-dimensional $(-1)$-dg symplectic space. The topological action is the free one:
\bea S\lsb \varphi\rsb &\coloneqq \int_{S^1} \lan \varphi, d\varphi\ran_\omega\\
&= \sum_i \int_{S^1} \bP_i d\bQ^i= \sum_i \int_{S^1} p_i(\theta) dq^i(\theta). \eea

\begin{rmk}
This is the first-order formalism of topological quantum mechanics along the line of the \emph{AKSZ} construction \cite{AKSZ}. 
\end{rmk}

\subsubsection*{Propagator}
Let us choose the standard flat metric on $S^1$. Let $d^\ast$ be the adjoint of $d$. The Laplacian is 
\bea d d^\ast+d^\ast d=-\lb \frac{d}{d\theta}\rb^2.\eea

Let 
\bea\Pi=\omega^{-1}=\sum_i \frac{\p}{\p p_i} \wedge \frac{\p}{\p q^i}= \hf \sum_i \lb \frac{\p}{\p p_i} \otimes \frac{\p}{\p q^i}- \frac{\p}{\p q^i} \otimes\frac{\p}{\p p_i}\rb \ \in \asym^2 V\eea
be the Poisson bivector (or Poisson kernel).
Let 
\bea h_t\lb \theta_1, \theta_2\rb= \frac{1}{\sqrt{4\pi t}} \sum _{n\in \bZ} e^{-\frac{\lb \theta_1-\theta_2+n\rb^2}{4t}} \eea
be the standard heat kernel on $S^1$. Then the regularized propagator is 
\bea P^L_\varepsilon =\int^L_\varepsilon \p_{\theta_1} h_t\lb \theta_1, \theta_2\rb dt \otimes \Pi \ \in \cE\otimes \cE,\eea
where $\int^L_\varepsilon \p_{\theta_1} h_t\lb \theta_1, \theta_2\rb dt  \in C^\infty(S^1 \times S^1)$ and $\Pi\in V\otimes V$. Let us denote 
\bea P^{S^1} \lb \theta_1, \theta_2\rb= \int^\infty_0 \p_{\theta_1} h_t\lb \theta_1, \theta_2\rb dt.\eea
Then the full propagator is given by 
\bea P^\infty_0= P^{S^1} \otimes \Pi.\eea

\begin{prop}
$P^{S^1} \lb \theta_1, \theta_2\rb$ is the following periodic function of $\theta_1-\theta_2 \in \bR/\bZ$ 
\bea P^{S^1} \lb \theta_1, \theta_2\rb= \theta_1-\theta_2-\hf \quad \text{if } \quad 0<\theta_1-\theta_2 <1.\eea
\bea 
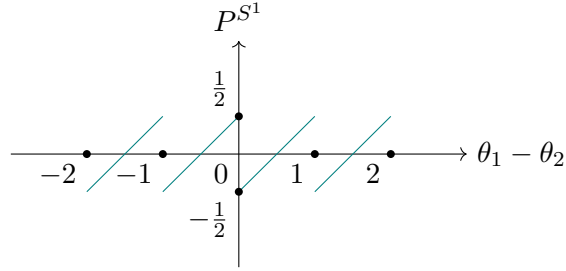
\begin{tikzpicture}
    \draw [->] (-3, 0) -- (3, 0) node [right] {$\theta_1-\theta_2$};
    \draw [->] (0, -1.5) -- (0, 1.5) node [above] {$P^{S^1}$};

    \node [anchor = north east] at (-2, 0) {$-2$};
    \node [anchor = north east] at (-1, 0) {$-1$};
    \node [anchor = north east] at (0, 0) {$0$};
    \node [anchor = north east] at (1, 0) {$1$};
    \node [anchor = north east] at (2, 0) {$2$};
    
    \foreach \x in {-2, -1, 0, 1} {
      \draw [teal] (\x, -0.5) -- (\x + 1, 0.5);
    }
    \node [anchor = south east] at (0, 0.5) {$\frac{1}{2}$};
    \node [anchor = north east] at (0, -0.5) {$-\frac{1}{2}$};
    
    \node [fill, circle, inner sep = 0, minimum size = 3] at (0, 0.5) {};
    \node [fill, circle, inner sep = 0, minimum size = 3] at (0, -0.5) {};
    \node [fill, circle, inner sep = 0, minimum size = 3] at (-2, 0) {};
    \node [fill, circle, inner sep = 0, minimum size = 3] at (-1, 0) {};
    \node [fill, circle, inner sep = 0, minimum size = 3] at (1, 0) {};
    \node [fill, circle, inner sep = 0, minimum size = 3] at (2, 0) {};
\end{tikzpicture}\eea
$P^{S^1}$ is NOT a smooth function on $S^1 \times S^1$ (as expected), but it is \emph{bounded}.
\end{prop}

\subsubsection*{Correlation map}
Let us denote the {formal Weyl algebra}
\bea {\cW_{2n}= \lb \bR [[ p_i,q^i]] ((\hbar)), \star\rb}\ ,\eea
and the {formal Weyl subalgebra}
\bea {\cW_{2n}^+ = \lb \bR [[ p_i,q^i]] [[\hbar]], \star\rb}\ ,\eea
where $\star$ is the Moyal product. We can identify the formal Weyl subalgebra as (formal) functions on $V$ (via deformation quantization):
\bea \cW_{2n}^+ \simeq \lb\widehat{\sO}(V)[[ \hbar]], \star\rb. \eea

Given $f_0,f_1,\cdots, f_m \in \cW_{2n}$, we define
$\cO_{f_0,f_1,\cdots, f_m} \in \sO(\cE)((\hbar))$
by
\bea \cO_{f_0,f_1,\cdots, f_m}\lsb \varphi\rsb\coloneqq \int_{0<\theta_1<\theta_2<\cdots<\theta_m<1} d\theta_1 d\theta_2 \cdots d\theta_m f^{(0)}_{0}\lb \varphi(\theta_0)\rb f^{(1)}_{1}\lb \varphi(\theta_1)\rb \cdots f^{(1)}_{m}\lb \varphi(\theta_m)\rb. \eea
Here $\varphi\in \Omega^\blt_{S^1}\otimes V$,  $f(\varphi(\theta))=f(\bP_i(\theta),\bQ^i(\theta))\in \Omega^\blt_{S^1}$, 
and we decompose it as 
$$
f(\varphi(\theta))=f^{(0)}(\varphi(\theta))+f^{(1)}(\varphi(\theta))d\theta.
$$
\begin{figure}[!htpb]\centering 
\tikzset{every picture/.style={line width=0.75pt}}         
\begin{tikzpicture}[x=0.75pt,y=0.75pt,yscale=-1,xscale=1]

\draw   (392.86,103.45) .. controls (392.86,73.34) and (417.27,48.93) .. (447.39,48.93) .. controls (477.5,48.93) and (501.91,73.34) .. (501.91,103.45) .. controls (501.91,133.57) and (477.5,157.98) .. (447.39,157.98) .. controls (417.27,157.98) and (392.86,133.57) .. (392.86,103.45) -- cycle ;
\draw  [fill={rgb, 255:red, 0; green, 0; blue, 0 }  ,fill opacity=1 ] (496,122.2) .. controls (496,120.41) and (497.45,118.96) .. (499.25,118.96) .. controls (501.04,118.96) and (502.49,120.41) .. (502.49,122.2) .. controls (502.49,124) and (501.04,125.45) .. (499.25,125.45) .. controls (497.45,125.45) and (496,124) .. (496,122.2) -- cycle ;
\draw  [fill={rgb, 255:red, 0; green, 0; blue, 0 }  ,fill opacity=1 ] (469.52,151.96) .. controls (469.52,150.16) and (470.98,148.71) .. (472.77,148.71) .. controls (474.57,148.71) and (476.02,150.16) .. (476.02,151.96) .. controls (476.02,153.75) and (474.57,155.21) .. (472.77,155.21) .. controls (470.98,155.21) and (469.52,153.75) .. (469.52,151.96) -- cycle ;
\draw  [fill={rgb, 255:red, 0; green, 0; blue, 0 }  ,fill opacity=1 ] (495.14,84.66) .. controls (495.14,82.87) and (496.6,81.41) .. (498.39,81.41) .. controls (500.18,81.41) and (501.64,82.87) .. (501.64,84.66) .. controls (501.64,86.45) and (500.18,87.91) .. (498.39,87.91) .. controls (496.6,87.91) and (495.14,86.45) .. (495.14,84.66) -- cycle ;
\draw  [fill={rgb, 255:red, 0; green, 0; blue, 0 }  ,fill opacity=1 ] (422.06,153.48) .. controls (422.06,151.68) and (423.51,150.23) .. (425.31,150.23) .. controls (427.1,150.23) and (428.55,151.68) .. (428.55,153.48) .. controls (428.55,155.27) and (427.1,156.73) .. (425.31,156.73) .. controls (423.51,156.73) and (422.06,155.27) .. (422.06,153.48) -- cycle ;

\draw (479.27,146.11) node [anchor=north west][inner sep=0.75pt]    {$f_{0}^{( 0)}$};
\draw (368.21,149.41) node [anchor=north west][inner sep=0.75pt]    {$f_{m}^{( 1)} d\theta _{m}$};
\draw (449.27,40) node  [rotate=-1.35]  {$\cdots $};
\draw (505.25,112.86) node [anchor=north west][inner sep=0.75pt]    {$f_{1}^{( 1)} d\theta _{1}$};
\draw (504.64,72.56) node [anchor=north west][inner sep=0.75pt]    {$f_{2}^{( 1)} d\theta _{2}$};
\draw (250.5,82.4) node [anchor=north west][inner sep=0.75pt]    {$\int _{\theta _{0} =0< \theta _{1} < \theta _{2} < \cdots < \theta _{m} < 1}$};
\end{tikzpicture}
\end{figure}

\begin{rmk}
$f^{(1)}(\varphi)$ is the {topological descent} of $f^{(0)}(\varphi)$ in topological  field theory.
\end{rmk}

Now let us apply the HRG flow 
$$
\exp{(\hbar P^\infty_0)}\lb \cO_{f_0,f_1,\cdots, f_m}\rb.
$$
Since $P^\infty_0$ is bounded, it is convergent and well-defined! This is the \emph{UV finite property}. As we have discussed, at $L=\infty$, we can view it as defining a function on zero modes  
\bea\bH =H^\blt\lb \Omega^\blt_{S^1} \otimes V, d\rb= H^\blt(S^1)\otimes V= V\oplus V d\theta.\eea
On zero modes we have $\widehat{\sO}(\bH)=\widehat{\Omega}^{-\blt}_{2n}$ forms on $V$.

\begin{defn}\label{defn-TQM-correlation}
We define the following correlation map:
\bea \lan \cdots\ran_{free}: \cW_{2n} \otimes \cdots \otimes \cW_{2n} \to \widehat{\Omega}^{-\blt}_{2n}((\hbar))\eea
by 
$$
\lan f_0\otimes f_1\otimes \cdots\otimes f_m\ran_{free} \coloneqq \left. \exp{(\hbar P^\infty_0)}\lb \cO_{f_0,f_1,\cdots, f_m}\rb \right|_{\bH}.
$$
In the path integral perspective, this is 
\bea \lan f_0\otimes f_1\otimes \cdots\otimes f_m\ran_{free}(\alpha)
=\int_{\Im d^\ast \subset \cE} \lsb D\varphi\rsb e^{-S[\varphi+\alpha]/\hbar} \cO_{f_0,f_1,\cdots,f_m}[\varphi+\alpha], \quad \alpha\in \bH= H^\blt(S^1)\otimes V.\eea
Here the zero mode $\alpha$ is viewed as the background field. It can also be represented as a Feynman diagram as follows.
\begin{figure}[!htpb]\centering 
\tikzset{every picture/.style={line width=0.75pt}}         
\begin{tikzpicture}[x=0.75pt,y=0.75pt,yscale=-1,xscale=1]
\draw   (370.33,193.5) .. controls (370.33,163.12) and (394.96,138.5) .. (425.33,138.5) .. controls (455.71,138.5) and (480.33,163.12) .. (480.33,193.5) .. controls (480.33,223.88) and (455.71,248.5) .. (425.33,248.5) .. controls (394.96,248.5) and (370.33,223.88) .. (370.33,193.5) -- cycle ;
\draw    (374.17,120.17) -- (389.67,152) ;
\draw [shift={(389.67,152)}, rotate = 64.04] [color={rgb, 255:red, 0; green, 0; blue, 0 }  ][fill={rgb, 255:red, 0; green, 0; blue, 0 }  ][line width=0.75]      (0, 0) circle [x radius= 3.35, y radius= 3.35]   ;
\draw    (347.33,170.17) -- (389.67,152) ;
\draw [shift={(389.67,152)}, rotate = 336.77] [color={rgb, 255:red, 0; green, 0; blue, 0 }  ][fill={rgb, 255:red, 0; green, 0; blue, 0 }  ][line width=0.75]      (0, 0) circle [x radius= 3.35, y radius= 3.35]   ;
\draw    (346.83,240.67) -- (377.17,219) ;
\draw [shift={(377.17,219)}, rotate = 324.46] [color={rgb, 255:red, 0; green, 0; blue, 0 }  ][fill={rgb, 255:red, 0; green, 0; blue, 0 }  ][line width=0.75]      (0, 0) circle [x radius= 3.35, y radius= 3.35]   ;
\draw    (403.67,267.17) -- (425.33,248.5) ;
\draw [shift={(425.33,248.5)}, rotate = 319.25] [color={rgb, 255:red, 0; green, 0; blue, 0 }  ][fill={rgb, 255:red, 0; green, 0; blue, 0 }  ][line width=0.75]      (0, 0) circle [x radius= 3.35, y radius= 3.35]   ;
\draw    (452.17,267.17) -- (425.33,248.5) ;
\draw [shift={(425.33,248.5)}, rotate = 214.82] [color={rgb, 255:red, 0; green, 0; blue, 0 }  ][fill={rgb, 255:red, 0; green, 0; blue, 0 }  ][line width=0.75]      (0, 0) circle [x radius= 3.35, y radius= 3.35]   ;
\draw    (492.17,250.67) -- (465.33,232) ;
\draw [shift={(465.33,232)}, rotate = 214.82] [color={rgb, 255:red, 0; green, 0; blue, 0 }  ][fill={rgb, 255:red, 0; green, 0; blue, 0 }  ][line width=0.75]      (0, 0) circle [x radius= 3.35, y radius= 3.35]   ;
\draw    (509.17,205.17) -- (480.33,186) ;
\draw [shift={(480.33,186)}, rotate = 213.61] [color={rgb, 255:red, 0; green, 0; blue, 0 }  ][fill={rgb, 255:red, 0; green, 0; blue, 0 }  ][line width=0.75]      (0, 0) circle [x radius= 3.35, y radius= 3.35]   ;
\draw    (506.67,162.67) -- (480.33,186) ;
\draw [shift={(480.33,186)}, rotate = 138.46] [color={rgb, 255:red, 0; green, 0; blue, 0 }  ][fill={rgb, 255:red, 0; green, 0; blue, 0 }  ][line width=0.75]      (0, 0) circle [x radius= 3.35, y radius= 3.35]   ;
\draw [color={rgb, 255:red, 144; green, 19; blue, 254 }  ,draw opacity=1 ]   (377.17,219) .. controls (407.67,213.67) and (424.17,238.17) .. (425.33,248.5) ;
\draw [color={rgb, 255:red, 144; green, 19; blue, 254 }  ,draw opacity=1 ]   (480.33,186) .. controls (461.5,196) and (457,212.5) .. (465.33,232) ;
\draw [color={rgb, 255:red, 144; green, 19; blue, 254 }  ,draw opacity=1 ]   (453.33,146.17) .. controls (488.33,142.67) and (485.83,172.67) .. (480.33,186) ;
\draw    (464.83,120.67) -- (453.33,146.17) ;
\draw [shift={(453.33,146.17)}, rotate = 114.27] [color={rgb, 255:red, 0; green, 0; blue, 0 }  ][fill={rgb, 255:red, 0; green, 0; blue, 0 }  ][line width=0.75]      (0, 0) circle [x radius= 3.35, y radius= 3.35]   ;
\draw [color={rgb, 255:red, 144; green, 19; blue, 254 }  ,draw opacity=1 ]   (389.67,152) -- (465.33,232) ;

\draw (400.67,202.4) node [anchor=north west][inner sep=0.75pt]  [color={rgb, 255:red, 144; green, 19; blue, 254 }  ,opacity=1 ]  {$P_{0}^{\infty }$};
\draw (410.17,156.9) node [anchor=north west][inner sep=0.75pt]  [color={rgb, 255:red, 144; green, 19; blue, 254 }  ,opacity=1 ]  {$P_{0}^{\infty }$};
\draw (448.67,176.9) node [anchor=north west][inner sep=0.75pt]  [color={rgb, 255:red, 144; green, 19; blue, 254 }  ,opacity=1 ]  {$P_{0}^{\infty }$};
\draw (478.67,134.9) node [anchor=north west][inner sep=0.75pt]  [color={rgb, 255:red, 144; green, 19; blue, 254 }  ,opacity=1 ]  {$P_{0}^{\infty }$};
\draw (334.33,234.9) node [anchor=north west][inner sep=0.75pt]    {$\alpha $};
\draw (334.33,159.9) node [anchor=north west][inner sep=0.75pt]    {$\alpha $};
\draw (365.83,104.9) node [anchor=north west][inner sep=0.75pt]    {$\alpha $};
\draw (391.33,260.9) node [anchor=north west][inner sep=0.75pt]    {$\alpha $};
\draw (451.33,261.9) node [anchor=north west][inner sep=0.75pt]    {$\alpha $};
\draw (493.83,244.9) node [anchor=north west][inner sep=0.75pt]    {$\alpha $};
\draw (510.33,197.4) node [anchor=north west][inner sep=0.75pt]    {$\alpha $};
\draw (509.33,150.9) node [anchor=north west][inner sep=0.75pt]    {$\alpha $};
\draw (463.33,103.4) node [anchor=north west][inner sep=0.75pt]    {$\alpha $};
\draw (353.34,208.55) node [anchor=north west][inner sep=0.75pt]  [rotate=-269.59]  {$\cdots $};
\draw (410.4,119.77) node [anchor=north west][inner sep=0.75pt]  [rotate=-0.52]  {$\cdots $};
\draw (420.83,253.4) node [anchor=north west][inner sep=0.75pt]    {$0$};
\draw (356.67,151.21) node [anchor=north west][inner sep=0.75pt]  [rotate=-302.24]  {$\cdots $};
\end{tikzpicture}
\end{figure}
\end{defn}

\begin{rmk}See \cite{LWY} for a probabilistic approach where the topological correlations above are constructed in terms of a large variance limit.
\end{rmk}

\subsubsection*{(Cyclic) Hochschild complex reviewed}
Let $A$ be a unital associative algebra and $\ols{A} \coloneqq A/(\bC\cdot 1)$. Let $C_{-p}(A)\coloneqq A\otimes \ols{A}^{\otimes p}$ be the cyclic $p$-chains. It carries a natural {Hochschild differential}
\bea b: C_{-p}(A) \to C_{-p+1}(A), \quad p\geq 1\eea
by 
$$
b(a_0\otimes \cdots\otimes a_p)=(-1)^{p} a_p a_0\otimes \cdots\otimes a_{p-1}+\sum_{i=0}^{p-1} (-1)^i a_0\otimes \cdots\otimes a_{i}a_{i+1} \otimes\cdots\otimes a_p. 
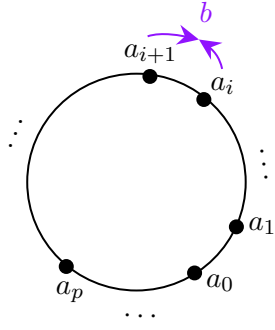
$$
\bea
\tikzset{every picture/.style={line width=0.75pt}}
\begin{tikzpicture}[x=0.75pt,y=0.75pt,yscale=-1,xscale=1]

\draw   (44.9,196.45) .. controls (44.9,166.34) and (69.31,141.93) .. (99.43,141.93) .. controls (129.54,141.93) and (153.95,166.34) .. (153.95,196.45) .. controls (153.95,226.57) and (129.54,250.98) .. (99.43,250.98) .. controls (69.31,250.98) and (44.9,226.57) .. (44.9,196.45) -- cycle ;
\draw  [fill={rgb, 255:red, 0; green, 0; blue, 0 }  ,fill opacity=1 ] (102.91,143.32) .. controls (102.91,141.52) and (104.36,140.07) .. (106.16,140.07) .. controls (107.95,140.07) and (109.4,141.52) .. (109.4,143.32) .. controls (109.4,145.11) and (107.95,146.57) .. (106.16,146.57) .. controls (104.36,146.57) and (102.91,145.11) .. (102.91,143.32) -- cycle ;
\draw  [fill={rgb, 255:red, 0; green, 0; blue, 0 }  ,fill opacity=1 ] (129.82,154.92) .. controls (129.82,153.13) and (131.28,151.67) .. (133.07,151.67) .. controls (134.86,151.67) and (136.32,153.13) .. (136.32,154.92) .. controls (136.32,156.71) and (134.86,158.17) .. (133.07,158.17) .. controls (131.28,158.17) and (129.82,156.71) .. (129.82,154.92) -- cycle ;
\draw  [fill={rgb, 255:red, 0; green, 0; blue, 0 }  ,fill opacity=1 ] (146.06,218.96) .. controls (146.06,217.16) and (147.52,215.71) .. (149.31,215.71) .. controls (151.11,215.71) and (152.56,217.16) .. (152.56,218.96) .. controls (152.56,220.75) and (151.11,222.21) .. (149.31,222.21) .. controls (147.52,222.21) and (146.06,220.75) .. (146.06,218.96) -- cycle ;
\draw  [fill={rgb, 255:red, 0; green, 0; blue, 0 }  ,fill opacity=1 ] (125.18,242.16) .. controls (125.18,240.37) and (126.64,238.91) .. (128.43,238.91) .. controls (130.22,238.91) and (131.68,240.37) .. (131.68,242.16) .. controls (131.68,243.95) and (130.22,245.41) .. (128.43,245.41) .. controls (126.64,245.41) and (125.18,243.95) .. (125.18,242.16) -- cycle ;
\draw  [fill={rgb, 255:red, 0; green, 0; blue, 0 }  ,fill opacity=1 ] (61.14,238.91) .. controls (61.14,237.12) and (62.6,235.66) .. (64.39,235.66) .. controls (66.19,235.66) and (67.64,237.12) .. (67.64,238.91) .. controls (67.64,240.71) and (66.19,242.16) .. (64.39,242.16) .. controls (62.6,242.16) and (61.14,240.71) .. (61.14,238.91) -- cycle ;
\draw [color={rgb, 255:red, 144; green, 19; blue, 254 }  ,draw opacity=1 ]   (142.2,139.53) .. controls (140.79,133.55) and (138.45,130.35) .. (133.01,126.4) ;
\draw [shift={(130.6,124.73)}, rotate = 33.69] [fill={rgb, 255:red, 144; green, 19; blue, 254 }  ,fill opacity=1 ][line width=0.08]  [draw opacity=0] (10.72,-5.15) -- (0,0) -- (10.72,5.15) -- (7.12,0) -- cycle    ;
\draw [color={rgb, 255:red, 144; green, 19; blue, 254 }  ,draw opacity=1 ]   (127.62,124.04) .. controls (121.65,122.66) and (112.06,120.78) .. (105,123.13) ;
\draw [shift={(130.6,124.73)}, rotate = 192.99] [fill={rgb, 255:red, 144; green, 19; blue, 254 }  ,fill opacity=1 ][line width=0.08]  [draw opacity=0] (10.72,-5.15) -- (0,0) -- (10.72,5.15) -- (7.12,0) -- cycle    ;

\draw (91.58,125.59) node [anchor=north west][inner sep=0.75pt]    {$a_{i+1}$};
\draw (134.89,142.11) node [anchor=north west][inner sep=0.75pt]    {$a_{i}$};
\draw (153.8,212.68) node [anchor=north west][inner sep=0.75pt]    {$a_{1}$};
\draw (132.67,241.45) node [anchor=north west][inner sep=0.75pt]    {$a_{0}$};
\draw (57.75,244.91) node [anchor=north west][inner sep=0.75pt]    {$a_{p}$};
\draw (129.39,104.14) node [anchor=north west][inner sep=0.75pt]  [color={rgb, 255:red, 144; green, 19; blue, 254 }  ,opacity=1 ]  {$b$};
\draw (40,171.15) node  [rotate=-292.57]  {$\cdots $};
\draw (102.65,264) node  [rotate=-1.35]  {$\cdots $};
\draw (163.33,186) node  [rotate=-261.64]  {$\cdots $};
\end{tikzpicture}
\eea
Then the associativity implies 
$b\circ b=0$. Thus,
$(C_{-\blt}(A), b)$ defines the {Hochschild chain complex}. We can also define the {Connes operator}:
\bea B: C_{-p}(A) \to C_{-p-1}(A)\eea
by 
$$
B(a_0\otimes \cdots\otimes a_p)=1\otimes a_0 \otimes\cdots\otimes a_p+\sum_{i=1}^p (-1)^{pi} 1\otimes a_i \otimes\cdots\otimes a_p\otimes a_0 \otimes\cdots\otimes a_{i-1}.
$$

We have the following relations:
\bea b^2=0, \quad B^2=0, \quad [b,B]=bB+Bb=0.\eea
Let $u$ be a formal variable of $\on{deg}=2$. Then $(b+uB)^2=0$.
This defines a complex
\bea CC^{per}_{-\blt}(A)=\lb C_{-\blt}(A)[u,u^{-1}], b+uB\rb,\eea
called the {periodic cyclic complex}. For a systematic reference, see Loday \cite{Loday}. 

\subsubsection*{Back to Correlation map}
It is not hard to see that
\bea \lan \cdots\ran_{free}: C_{-p}\lb \cW_{2n}\rb \to \widehat{\Omega}^{-p}_{2n}((\hbar)),\eea
i.e. $\lan f_0\otimes f_1\otimes \cdots\otimes f_p\ran_{free}$ is a $p$-form. Recall that $\widehat{\Omega}^{-\blt}_{2n}$ is equipped with a BV operator $\Delta=\cL_{\omega^{-1}}=\cL_\Pi$, the Lie derivative with respect to the Poisson bi-vector.

\begin{prop}[\cite{GLX-Effective}]
\bea \lan b(-)\ran_{free}= \hbar \Delta \lan \cdots\ran_{free},\\
\lan B(-)\ran_{free}= d_{2n} \lan \cdots\ran_{free}.\eea
Here $d_{2n}: \widehat{\Omega}^{-\blt}_{2n} \to \widehat{\Omega}^{-(\blt+1)}_{2n}$ is the de Rham differential. 
\end{prop}
In other words, the {correlation map}:
\bea  \lan \cdots\ran_{free}: C_{-\blt}(\cW_{2n}) \to \widehat{\Omega}^{-\blt}_{2n}((\hbar ))\eea
intertwines $b$ with $\hbar\Delta$ and $B$ with $d_{2n}$. We can combine the above two to get
\bea \lan\cdots\ran_{free}: CC^{per}_{-\blt}(\cW_{2n})\to 
\widehat{\Omega}^{-\blt}_{2n}((\hbar))[u,u^{-1}]\eea
which intertwines $b+uB$ with $\hbar\Delta+ud_{2n}$.

\subsubsection*{BV integral on zero modes}
We can define a \emph{BV integration} map on the BV algebra $\lb \widehat{\Omega}^{-\blt}_{2n},\Delta\rb$
which is only non-zero on top forms $\widehat{\Omega}^{-2n}_{2n}$ and sends
\bea \beta\in \widehat{\Omega}^{-2n}_{2n} \mapsto \left. \frac{\hbar^n}{n!}\iota^n_\Pi \beta \right|_{p=q=0}.\eea
This is the {Berezin integral} \cite{Berezin-phase} over the purely fermionic super Lagrangian. We can extend this BV integration to an $S^1$-equivariant version by
\bea \int_{BV}: \widehat{\Omega}^{-\blt}_{2n}[u,u^{-1}] \to \bR((\hbar))[u,u^{-1}], \quad \beta\mapsto \left. \lb u^n e^{\hbar\iota_\Pi/u}\beta\rb\right|_{p=q=0}.\eea
Then it has the following property
\bea \int_{BV} \lb \hbar\Delta +ud_{2n}\rb(-)=0.\eea

\begin{rmk}
For $\beta\in \widehat{\Omega}^{-\blt}_{2n}$, the non-equivariant limit
\bea \lim_{u\to 0} \int_{BV} \beta= \left. \frac{\hbar^n}{n!}\iota_\Pi^n \beta\right|_{p=q=0}\eea
gives back the Berezin integral.
\end{rmk}

Combining the above maps, we define
\bea \Tr \coloneqq \int_{BV}  \circ \lan \cdots\ran_{free}: CC^{per}_{-\blt}(\cW_{2n})\to \bR((\hbar))[u,u^{-1}]\eea
which satisfies the following equation:
\bea\Tr\lb (b+uB)(-)\rb=0.\eea
Therefore $\Tr$ descends to {periodic cyclic homology}. This essentially leads to the trace formula in {Feigin-Felder-Shoikhet} \cite{FFS}.

\subsubsection*{Quantum Master Equation}
We can generalize slightly by considering a graded vector space $V$ with a $\on{deg}=0$ symplectic pairing $\omega$. We still have the canonical quantization $\lb \widehat{\sO}(V)[[ \hbar]],\star\rb$ by Moayl product and similarly can define the BV algebra of forms 
$$
\lb \widehat{\Omega}^{-\blt}_{V},\Delta=\cL_{\omega^{-1}}\rb.
$$
The same trace map gives 
\bea \lan \cdots\ran_{free}: C_{-\blt}\lb \widehat{\sO}(V)[[\hbar]]\rb\to 
\widehat{\Omega}^{-\blt}_{V}((\hbar)), \quad b\mapsto \hbar\Delta.\eea

Given $\gamma\in \widehat{\sO}(V)[[\hbar]]$, $\on{deg}(\gamma)=1$, it defines an action functional:
\bea I_\gamma= \int_{S^1}\gamma(\varphi) \quad \forall \varphi\in \Omega^\blt(S^1)\otimes V.\eea
Let us treat $I_\gamma$ as an {interaction} and consider
\bea \underbrace{\hf \int_{S^1} \lan \varphi,d\varphi\ran}_{\text{free part}}
+\underbrace{\int_{S^1} \gamma(\varphi)}_{I_\gamma}.\eea
Then we run the HRG flow to get
\bea e^{\frac{1}{\hbar}I_\gamma[\infty]}\coloneqq e^{\hbar\partial_{P^\infty_0}} e^{\frac{1}{\hbar}I_\gamma}\eea
which is well-defined since $P^\infty_0$ is bounded.

Let us now analyze the QME. By construction,
\bea e^{\frac{1}{\hbar}I_\gamma[\infty]}= \lan 1\otimes e^{\gamma/\hbar}\ran_{free}.\eea
Assume $\gamma\star \gamma=\hf \lsb \gamma,\gamma\rsb_\star=0$. Then 
\bea \hbar\Delta e^{\frac{1}{\hbar}I_\gamma[\infty]}= \lan b\lb 1\otimes e^{\gamma/\hbar}\rb \ran_{free}=0.\eea

\begin{prop}[\cite{GLL-BV}]
If $\lsb \gamma,\gamma\rsb_\star=0$, then the local interaction $I_\gamma= \int_{S^1} \gamma(\varphi)$ defines a family of solutions of effective QME $I_\gamma[L]$ at scale $L>0$ by
\bea e^{\frac{1}{\hbar}I_\gamma[L]}\coloneqq \lim_{\varepsilon\to\infty} e^{\hbar \partial_{P^L_\varepsilon}} e^{\frac{1}{\hbar}I_\gamma}.\eea
\end{prop}

\subsection{Global Theory}
\label{sec:tqm2}

Recall in Section \ref{sec:tqm1}, we have discussed the first-order formalism of TQM such that in a local model with maps $\varphi: \Omega^\blt_{S^1} \to V\simeq \bR^{2n}$, the correlation map
\bea  \lan \cdots\ran_{free}: C_{-\blt}(\cW_{2n}) \to \widehat{\Omega}^{-\blt}_{2n}((\hbar ))\eea
intertwines $b$ with $\hbar\Delta$ and $B$ with $d_{2n}$.

In this section, we are going to glue this construction to a symplectic manifold and establish the algebraic index to universal Lie algebra
cohomology computations. The basic idea is to \emph{glue} the local model $\Sigma\to T^{Model}\subset X$. In the following discussion,  we borrow the presentation in \cite{GLX-Effective}, where extensive references are given for related material.

\subsubsection*{Gluing via Gelfand-Kazhdan formal geometry}
\begin{defn}
A \textbf{Harish-Chandra pair} is a pair $(\fg,K)$, where $\fg$ is a Lie algebra, $K$ is a Lie group, with 
\begin{itemize}
    \item an action of $K$ on $\fg$: $K\xrightarrow{\ \rho\ } \on{Aut}(\fg)$,
    \item a natural embedding: $\on{Lie}(K) \xhookrightarrow{\ i\ } \fg$, where $\on{Lie}(K)$ is the Lie algebra associated with $K$,
\end{itemize}
such that they are compatible:
\bea
\begin{tikzcd}
\on{Lie}(K) \ar[r, hook, "i"] \ar[dr, "d\rho"]
& \fg \ar[d,"adjoint"] \\
& \on{Der}(\fg)
\end{tikzcd}
\eea
\end{defn}

\begin{defn}
A \textbf{$(\fg,K)$-module} is a vector space $V$ with 
\begin{itemize}
    \item an action of $K$ on $V$: $K \xrightarrow{\ \varphi\ } \on{GL}(V)$,
    \item a Lie algebra morphism: $\fg\to \on{End}(V)$,
\end{itemize}
such that they are compatible:
\bea
\begin{tikzcd}
\on{Lie}(K) \ar[r, hook, "i"] \ar[dr, "d\varphi"]
& \fg \ar[d] \\
& \on{End}(V)
\end{tikzcd}
\eea
\end{defn}

\begin{defn}
A \textbf{flat $(\fg,K)$-bundle} over X is
\begin{itemize}
    \item a principal $K$-bundle $P \xrightarrow{\ \pi\ } X$,
    \item a $K$-equivariant $\fg$-valued 1-form $\gamma\in \Omega^1(P,\fg)$ on $P$,
\end{itemize}
satisfying the following conditions:
\bi[(1)]
    \item $\forall a\in \on{Lie}(K)$, let $\xi_a\in \on{Vect}(P)$ generated by $a$. Then we have the contraction $\gamma(\xi_a)=a$ such that
    \bea
    \begin{tikzcd}
        0 \ar[r] & \on{Lie}(K) \ar[r] \ar[dr, hook, "i"] & \on{Vect}(P) \ar[d, "\gamma"] \\
        & & \fg
    \end{tikzcd}
    \eea
    \item $\gamma$ satisfies the Maurer-Cartan equation
    \bea d\gamma +\hf \lsb \gamma,\gamma\rsb=0,\eea
    where $d$ is the de Rham differential on $P$, and $\lsb-,-\rsb$ is the Lie bracket in $\fg$.
\ei
\end{defn}

Given a flat $(\fg,K)$-bundle $P\to X$ and $(\fg,K)$-module $V$, let \bea\Omega^\blt(P,V)\coloneqq \Omega^\blt (P)\otimes V\eea
denote differential forms on $P$ valued in $V$.
It carries a connection
\bea \nabla^\gamma= d+\gamma :  \Omega^\blt(P,V)\to \Omega^{\blt+1}(P,V)\eea
which is \emph{flat} by the Maurer-Cartan equation.
The group $K$ acts on $\Omega^\blt(P)$ and $V$, and hence inducing a natural action on $\Omega^\blt(P,V)$. Let 
\bea V_P\coloneqq P \times_K V\eea 
be the vector bundle on $X$ associated to the $K$-representation $V$. Let 
$$
\Omega^\blt(X;V_P)
$$ 
be differential forms
on $X$ valued in the bundle $V_P\to X$.
Similar to the usual principal bundle case, $\nabla^\gamma$ induces a flat connection on $V_P\to X$. This defines a (de Rham) chain complex $\lb \Omega^\blt(X;V_P), \nabla^\gamma\rb$, and $H^\blt(X;V_P)$ denotes the corresponding de Rham cohomology.

We can descend Lie algebra cohomologies to geometric objects on $X$.
\begin{defn}
Let $V$ be a $(\fg,K)$-module. Define the \textbf{$(\fg,K)$ relative Lie algebra cochain complex} 
$\lb C^\blt_{\Lie}(\fg,K;V), \p_{Lie}\rb$
by
\bea C^p_{\Lie}(\fg,K;V)=\on{Hom}_K \lb \asym^p\lb\fg/\on{Lie}(K)\rb,V\rb.\eea
\end{defn}
Here $\on{Hom}_K$ means $K$-equivariant linear maps. $\p_{\Lie}$ is the Chevalley-Eilenberg differential if we view
$C^p_{\Lie}(\fg,K;V)$ as a subspace of the Lie algebra cochain $C^p_{\Lie}(\fg;V)$. Explicitly,
for $\alpha\in C^p_{\Lie}(\fg,K;V)$,
\bea \lb \p_{\Lie}\alpha\rb \lb a_1 \wedge \cdots \wedge a_{p+1}\rb
&=\sum_{i=1}^{p+1}(-1)^{i-1} a_i\cdot \alpha\lb a_1 \wedge \cdots \wedge \widehat{a}_i \wedge \cdots \wedge a_{p+1}\rb\\
&\quad +\sum_{i<j} (-1)^{i+j} \alpha\lb \lsb a_i,a_j\rsb \wedge \cdots \wedge \widehat{a}_i \wedge \cdots \wedge \widehat{a}_j \wedge \cdots \wedge a_{p+1}\rb.\eea
The corresponding cohomology is $H^\blt_{\Lie}(\fg,K;V)$. 

Given a $(\fg,K)$-module $V$ and flat $(\fg,K)$-bundle $P\to X$ with the flat connection $\gamma\in\Omega^1(P,\fg)$. 
We can define the \textbf{descent map} from the $(\fg, K)$ relative Lie algebra cochain complex to
$V$-valued de Rham complex on $P$ by
\bea \on{desc}: \lb C^\blt_{Lie}(\fg,K;V),\p_{Lie}\rb\to \lb \Omega^\blt(X;V_P),\nabla^\gamma\rb, \quad \alpha\mapsto \alpha(\gamma,\cdots,\gamma)\eea
inducing the cohomology descent map
\bea \on{desc}: H^\blt_{Lie}(\fg,K;V)\to H^\blt (X;V_P).\eea

\subsubsection*{Fedosov connection revisited}
Recall the (formal) Weyl algebras
\bea \cW_{2n}=\bR[[p_i,q^i]] ((\hbar)), \quad 
\cW_{2n}^+=\bR[[p_i,q^i]] [[\hbar]]\eea
with the induced Lie algebra structure such that the Lie bracket is defined by
\bea \lsb f,g\rsb\coloneqq \frac{1}{\hbar} \lsb f,g\rsb_\star=\frac{1}{\hbar}\lb f\star g- g\star f\rb.\eea

Let $\gsp_{2n}$ be the symplectic group of linear transformations preserving the Poisson bivector $\Pi$. It acts on Weyl algebras by inner automorphisms. We can identify the Lie algebra $\sp_{2n}$ of $\gsp_{2n}$ with quadratic polynomials in $\bR[p_i,q^i]$, and $\sp_{2n}$ is a Lie subalgebra of $\cW_{2n}^+$.
The action $\gsp_{2n}\curvearrowright \bR^{2n}$ induces the action $\gsp_{2n}\curvearrowright \cW_{2n}^+$. Hence, $\lb \cW_{2n}^+,\gsp_{2n}\rb$ and $\lb \cW_{2n},\gsp_{2n}\rb$ are Harish-Chandra pairs.

Let $(X,\omega)$ be a symplectic manifold, and $F_{\gsp}(X)$ be the symplectic frame bundle. We have the Weyl bundles
\bea \cW^+_X=F_{\gsp}(X)\times_{\gsp_{2n}} \cW_{2n}^+, \quad 
\cW_X=F_{\gsp}(X) \times_{\gsp_{2n}} \cW_{2n}.\eea
Consider the Harish-Chandra pair
\bea (\ols{\fg},K)=(\fg/Z(\fg),\gsp_{2n}),\eea
where $\fg=\cW_{2n}^+$, and $Z(\fg)=\bR[[\hbar]]$ is the center of $\fg$, $Z(\fg)\cap \sp_{2n}=0$.
Fedosov constructed a flat $(\ols{\fg},K)$-bundle $F_{\gsp}(X)\to X$
and $H^0(X; \cW_X^+)$ gives a \emph{deformation quantization}. 
Choose the trivial $(\ols{\fg},K)$-module $\bR((\hbar))$. Then 
\bea \on{desc}: C^\blt_{\Lie} \lb \ols{\cW^+_{2n}},\sp_{2n}; \bR((\hbar))\rb \to \Omega^\blt_X((\hbar)).\eea 
This is the \textbf{Gelfand-Fuks map}.
Here 
$$
C^\blt_{\Lie} \lb \ols{\cW^+_{2n}},\sp_{2n}; \bR((\hbar))\rb \simeq C^\blt_{Lie}\lb \cW^+_{2n},\sp_{2n}\oplus Z(\cW^+_{2n}); \bR((\hbar))\rb.
$$

\subsubsection*{Characteristic classes}
Let us review the Chern-Weil construction of {characteristic classes} in Lie algebra cohomology.
They will descend to the usual characteristic forms via the Gelfand-Fuks map.

Let $\fg$ be a Lie algebra, and $\fh\subset \fg$ be its Lie subalgebra. Let the projection map
\bea \on{pr}: \fg\to\fh\eea
be the $\fh$-equivariant splitting of the embedding $\fh\subset \fg$. In general $\on{pr}$ is not a Lie algebra homomorphism from $\fg$ to $\fh$. The \emph{failure} of $\on{pr}$ being a Lie algebra homomorphism gives $R\in \on{Hom}\lb \asym^2\fg, \fh\rb$ by
\bea R(\alpha,\beta)=\lsb \on{pr}(\alpha),\on{pr}(\beta)\rsb_\fh-\on{pr}\lsb \alpha,\beta\rsb_\fg, \quad \alpha,\beta\in\fg.\eea
The $\fh$-equivariance of $\on{pr}$ implies that $R\in\on{Hom}_\fh\lb \asym^2(\fg/\fh),\fh\rb$. $R$ is called the \textbf{curvature form}.
Let $\sym^m(\fh^\vee)^\fh$ be $\fh$-invariant polynomials on $\fh$ of homogeneous degree $\on{deg}=m$. Given $P\in \sym^m(\fh^\vee)^\fh$, 
we can
associate a cochain 
$$
P(R)\in C^{2m}_{\Lie}(\fg, \fh; \bR)
$$
by the composition
\bea P(R): \asym^{2m}\fg \xrightarrow{\asym^m R} \sym^m(\fh) \xrightarrow{P}\bR.\eea
It can be checked that $\p_{\Lie} P(R)=0$, defining a cohomology class 
$$
[P(R)] \in H^{2m}(\fg,\fh;\bR)
$$ 
which does
not depend on the choice of $\on{pr}$. Therefore we have the analogue of Chern-Weil characteristic map
\bea \chi:\sym^\blt(\fh^\vee)^\fh \to H^\blt(\fg,\fh;\bR), \quad P\mapsto \chi(P) \coloneqq \lsb P(R)\rsb.\eea

Now we apply the above construction to the case where
\bea \fg= \cW^+_{2n}, \quad \fh=\sp_{2n}\oplus Z(\fg).\eea
Any element $f$ in $\fg= \cW^+_{2n}$ can be uniquely written as a polynomial
$f=f(y^i,\hbar)$, with coordinates $(y^1,\cdots,y^n,y^{n+1},\cdots, y^{2n})=(p_1,\cdots,p_n, q^1, \cdots, q^n)$. 
Define the $\fh$-equivariant projections
\bea \on{pr}_1(f) &=\left. \hf \sum_{i,j}\p_i \p_j f\right|_{y=\hbar=0} y^i y^j \in\sp_{2n},\\
\on{pr}_3(f) &=\left. f\right|_{y=0} \in Z(\fg).\eea
We obtain the corresponding curvature
\bea
R_1 &\coloneqq \lsb \on{pr}_1(-), \on{pr}_1(-)\rsb-\on{pr}_1\lsb-,-\rsb\ \in\on{Hom}(\asym^2 \fg,\sp_{2n}),\\
R_3 &\coloneqq -\on{pr}_3\lsb -,-\rsb\ \in\on{Hom}(\asym^2 \fg,\bR[[\hbar]]).
\eea

\begin{rmk}
A more general case can be considered when we incorporate vector bundles, where $\fg= \cW^+_{2n}+\hbar\lb \mathfrak{gl}\lb \cW^+_{2n}\rb\rb$, $\fh=\sp_{2n}\oplus \hbar\mathfrak{gl}\oplus Z(\fg)$. There the extra projection $\on{pr}_2$ and its corresponding curvature $R_2$ are defined as elements in $\hbar\mathfrak{gl}$ and $\on{Hom}\lb\asym^2, \mathfrak{gl}\rb$, respectively.
It is worthwhile to point out that all the $\on{Hom}$’s here are only $\bR$-linear map, but not
$\bR[[\hbar]]$-linear, although $\fg$ is a $\bR[[\hbar]]$-module.
\end{rmk}

We now define the \textbf{$\widehat{A}$-genus}
\bea\widehat{A}(\sp_{2n})\coloneqq \lsb \on{det}\lb \frac{R_1/2}{\sinh{(R_1/2)}}\rb^{\hf}\rsb \in H^\blt(\fg,\fh;\bR).\eea
Under the descent map $\on{desc}: H^\blt(\fg,\fh;\bR((\hbar)))\to H^\blt(X)((\hbar))$ via the Fedosov connection, it can be shown that
\bea \on{desc}\lb \widehat{A}(\sp_{2n})\rb &= \widehat{A}(X),\\
\on{desc}\lb R_3\rb &= \omega_\hbar- \hbar\omega.
\eea

\subsubsection*{Universal trace map}
Recall that using $\Omega^\blt_{S^1}\to \bR^{2n}$, we have obtained
\bea\Tr= \int_{BV}\circ \lan-\ran_{free}: CC^{per}_{-\blt}(\cW_{2n})\to \bK\coloneqq\bR((\hbar))[u,u^{-1}].\eea

Let us write 
\bea \Tr\in \on{Hom}_{\bK}\lb CC^{per}_{-\blt}(\cW_{2n}),\bK\rb.\eea
This is a $\lb {\cW^+_{2n}},\gsp_{2n}\rb$-module. Via the flat $\lb {\cW^+_{2n}},\gsp_{2n}\rb$-bundle $F_{\gsp}(X)\to X$, we obtain the associated bundle
\bea E^{per}\coloneqq F_{\gsp}(X) \times_{\gsp_{2n}}
\on{Hom}_{\bK}\lb CC^{per}_{-\blt}(\cW_{2n}),\bK\rb\eea
with induced flat connection $\nabla^\gamma$.

Recall the Weyl bundle $\cW(X)=F_{\gsp}(X) \times_{\gsp_{2n}} \cW_{2n}$ with flat connection $\nabla^\gamma$.
We would like to glue $\Tr$ on $X$. Let us denote $\delta$ for the differential on $\on{Hom}_{\bK}\lb CC^{per}_{-\blt}(\cW_{2n}),\bK\rb$ induced from $b+uB$. So \bea \delta \Tr=\Tr \lb (b+uB)(-)\rb=0.\eea

We can view $\Tr$ as defining an element in
\bea C^0_{Lie}\lb \fg,\fh; \on{Hom}_{\bK}\lb CC^{per}_{-\blt}(\cW_{2n}),\bK\rb\rb,\eea
where we take 
\bea\fg= \cW^+_{2n}/Z(\cW^+_{2n}), \quad \fh=\sp_{2n}.\eea
However, $\Tr$ is NOT $\fg$-invariant, i.e. $\p_{\Lie}\Tr\neq 0$. In other words, $\Tr$ is NOT a map of $\lb\fg,\gsp_{2n}\rb$-module. So $\Tr$ can not be glued directly. 

It is observed that $\p_{\Lie}\Tr=\delta(-)$. 
It turns out that we have a canonical way to lift $\Tr$ to
\bea \widehat{\Tr}\in C^\blt_{Lie}\lb \fg,\fh; \on{Hom}_{\bK}\lb CC^{per}_{-\blt}(\cW_{2n}),\bK\rb\rb\eea
such that
\bea\widehat{\Tr}=\Tr+ \text{terms in } C^{>0}_{Lie}\lb \fg,\fh; \on{Hom}_{\bK}\lb CC^{per}_{-\blt}(\cW_{2n}),\bK\rb\rb\eea
and satisfying the coupled cocycle condition
\bea \lb\p_{Lie}+\delta\rb \widehat{\Tr}=0.\eea
$\widehat{\Tr}$ is called the \textbf{universal trace map}. Let us insert $1\in \cW_{2n}$, then $\widehat{\Tr}(1)$ is $\p_{Lie}$-closed, which defines the \textbf{universal index}, $\lsb \widehat{\Tr}(1)\rsb\in H^\blt_{Lie}(\fg,\fh;\bK)$.

\begin{thm}[Universal algebraic index theorem]
\bea\lsb \widehat{\Tr}(1)\rsb=u^ne^{-R_3/(u\hbar)}\widehat{A}(\sp_{2n})_u,\eea
where for $A=\sum_{p \text{ even}}A_p$, $A_P\in H^p(\fg,\fh;\bK)$, \bea A_u=\sum_p u^{-p/2}A_p.\eea
\end{thm}
This theorem is developed in the works of Feigin-Tsygan \cite{feigin1989riemann}, Feigin-Felder-Shoikhet \cite{FFS}, Bressler-Nest-Tsygan \cite{bressler2002riemann}, and many others. This can be naturally generalized to the bundle case \cite{GLX-Effective} (as well as an explicit formula as a byproduct).

Now we apply the Gelfand-Fuks (descent) map on $\widehat{\Tr}$, such that
\bea\begin{tikzcd}
    C^\blt_{Lie}\lb \fg,\fh; \on{Hom}_{\bK}\lb CC^{per}_{-\blt}(\cW_{2n}),\bK\rb\rb \ar[d, "\on{desc}"]\\
    \Omega^\blt\lb X, \on{Hom}_{\bK}\lb CC^{per}_{-\blt}(\cW(X)),\bK\rb\rb
\end{tikzcd}\eea
Let $\cW_D(X)$ be the space of flat sections of $\cW(X)$
that gives a \emph{deformation quantization}. Then 
\bea
\on{desc}(\widehat{\Tr}): CC^{per}_{-\blt}(\cW_D(X))\to \Omega^\blt(X)((\hbar))[u,u^{-1}], \quad b+uB\mapsto d_X.
\eea
In particular, it defines a \emph{trace map} in deformation quantization by
\bea
f\in \cW_D(X)\mapsto \int_X \on{desc}(\widehat{\Tr})(f) \in \bR((\hbar)).
\eea

We can show that $\int_X \on{desc}(\widehat{\Tr})(f)$ does not depend on $u$. By the \emph{universal algebraic index theorem}, we have
\bea \int_X \on{desc}(\widehat{\Tr})(1)=\int_X e^{-\omega_\hbar/\hbar}\widehat{A}(X).\eea
This gives the {algebraic index theorem}.

\subsubsection*{Construction of universal trace map $\widehat{\Tr}$}

We have the following relations \cite{GLX-Effective}.
\bea
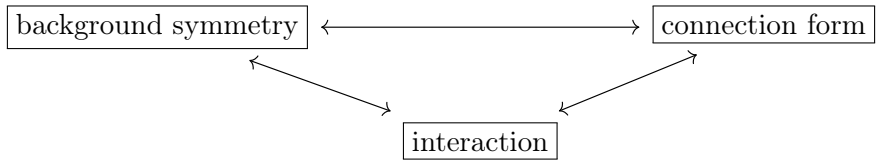
\begin{tikzcd}
    \boxed{\text{background symmetry}} \ar[rr, leftrightarrow] \ar[dr,leftrightarrow] & & \boxed{\text{connection form}} \ar[dl,leftrightarrow]\\
    & \boxed{\text{interaction}} &
\end{tikzcd}\eea

Let $\Theta: \fg\to \cW^+_{2n}/ Z(\cW^+_{2n})=\fg$
be the canonical identity map.
For each $f\in \cW^+_{2n}/ Z(\cW^+_{2n})$, we have defined the local functional on $\cE=\Omega^\blt(S^1)\otimes \bR^{2n}$ by
\bea I_f(\varphi)=\int_{S^1} f(\varphi), \quad \varphi\in\cE.\eea
Then $\Theta$ gives a map
\bea I_\Theta: \fg\to \cO_{loc}(\cE), \quad f\mapsto I_{\Theta(f)}.\eea

We can view this map as 
\bea I_\Theta\in C^1(\fg,\cO_{loc}(\cE))=\fg^\vee\otimes \cO_{loc}(\cE).\eea
This allows constructing 
$\widehat{\Tr}\in C^\blt_{Lie}\lb \fg,\fh; \on{Hom}_{\bK}\lb CC^{per}_{-\blt}(\cW_{2n}),\bK\rb\rb$ explicitly \cite{GLX-Effective} by
\bea
\widehat{\Tr}\lb f_0\otimes f_1\otimes \cdots\otimes f_m\rb &\coloneqq
\int_{BV} \exp{\lb\hbar P^\infty_0\rb}
\lb \cO_{f_0,f_1,\cdots,f_m} 
e^{\frac{1}{\hbar}I_\Theta}\rb\ \in C^\blt(\fg,\fh;\bK), \quad f_i\in \cW_{2n}\\
&``=\int_{BV}\int_{\Im d^\ast\subset\cE} e^{-\frac{1}{2\hbar} \int_{S^1} \lan\varphi,d\varphi\ran+\frac{1}{\hbar}I_\Theta} \cO_{f_0,f_1,\cdots,f_m}\,''.
\eea

\subsubsection*{Computation of index}
The Weyl algebra $\cW_{2n}$ can be viewed as a family of associative algebras parameterized by $\hbar$.
This leads to the {Gauss-Manin-Getzler  connection} \cite{Getzler} $\nabla_{\hbar \p_\hbar}$ on  $CC^{per}_{-\blt}(\cW_{2n})$.
The calculation of index consists of the following steps \cite{GLX-Effective}:
\bi[(1)]
\item \emph{Feynman diagram computation} implies
\bea \widehat{\Tr}(1)=u^n e^{-R_3/(u\hbar)}\lb\underbrace{\widehat{A}(\sp_{2n})_u}_{1-\text{loop computation}} +\cO(\hbar)\rb.\eea
\item Computation of \emph{Gauss-Manin-Getzler connection} shows 
$$
\nabla_{\hbar \p_\hbar} \lb e^{R_3/(u\hbar)}\widehat{
\Tr}(1)\rb\  \text{is $\p_{Lie}$-exact}.
$$
\item Combining (1) and (2), we find 
\bea \lsb \widehat{
\Tr}(1)\rsb= \lsb u^n e^{-R_3/(u\hbar)} \widehat{A}(\sp_{2n})_u\rsb\ \in H^\blt(\fg,\fh;\bK).\eea
\ei

\section{Two-dimensional Chiral Theory}
\label{sec:2d1}

We have discussed the first-order formalism of topological QM, where the fields are differential forms $\Omega^\blt(S^1,V)$ on $S^1$ valued in the vector bundle $V$ with the de Rham differential $d$. Here $d$ being part of the BRST operator implies that ``translation is homologically trivial.'' This defines a topological theory.

We will now consider 2d chiral models where the fields are differential forms $\Omega^{0,\blt}(\Sigma,h)$ with the Dolbeault differential $\ols{\p}$. The Dolbeault differential being part of the BRST operator implies that ``anti-holomorphic translation is homologically trivial,'' which in turn defines a chiral (or holomorphic) theory.

In topological QM, the theory is \emph{UV finite}. The general consideration in Section \ref{sec:UV} applies and we find that the renormalized QME is traded to Moyal commutator and Fedosov equation. We will see that 2d chiral theory is also \emph{UV finite} and we have a similar geometric result for QME \cite{LS-VOA}.

\subsection{Vertex Algebra}

As illustrated by the picture below, in 1d topological theory we have an associative algebra defined by the fusion of two operators $a\cdot b$; in 2d chiral theory we have (chiral) vertex algebra defined by $A_{(n)}B$. The algebras are found when one operator approaches another either on a line (for 1d) or on a plane (for 2d).
\bea 
\tikzset{every picture/.style={line width=0.75pt}} 
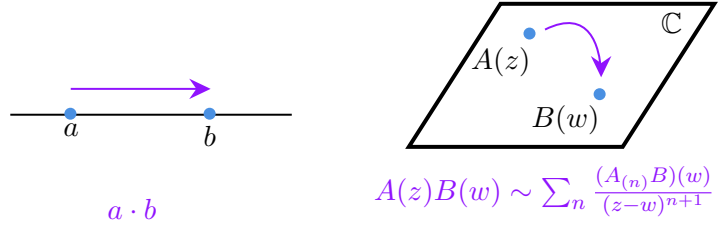
\begin{tikzpicture}[x=0.75pt,y=0.75pt,yscale=-1,xscale=1]

\draw    (35.58,72.83) -- (176.08,72.83) ;
\draw  [color={rgb, 255:red, 74; green, 144; blue, 226 }  ,draw opacity=1 ][fill={rgb, 255:red, 74; green, 144; blue, 226 }  ,fill opacity=1 ] (63.08,72.33) .. controls (63.08,70.95) and (64.2,69.83) .. (65.58,69.83) .. controls (66.96,69.83) and (68.08,70.95) .. (68.08,72.33) .. controls (68.08,73.71) and (66.96,74.83) .. (65.58,74.83) .. controls (64.2,74.83) and (63.08,73.71) .. (63.08,72.33) -- cycle ;
\draw  [color={rgb, 255:red, 74; green, 144; blue, 226 }  ,draw opacity=1 ][fill={rgb, 255:red, 74; green, 144; blue, 226 }  ,fill opacity=1 ] (132.58,72.33) .. controls (132.58,70.95) and (133.7,69.83) .. (135.08,69.83) .. controls (136.46,69.83) and (137.58,70.95) .. (137.58,72.33) .. controls (137.58,73.71) and (136.46,74.83) .. (135.08,74.83) .. controls (133.7,74.83) and (132.58,73.71) .. (132.58,72.33) -- cycle ;
\draw [color={rgb, 255:red, 144; green, 19; blue, 254 }  ,draw opacity=1 ]   (65.75,59.83) -- (132.58,59.83) ;
\draw [shift={(135.58,59.83)}, rotate = 180] [fill={rgb, 255:red, 144; green, 19; blue, 254 }  ,fill opacity=1 ][line width=0.08]  [draw opacity=0] (10.72,-5.15) -- (0,0) -- (10.72,5.15) -- (7.12,0) -- cycle    ;
\draw  [line width=1.5]  (280.25,17.04) -- (387,17.04) -- (341.25,89) -- (234.5,89) -- cycle ;
\draw  [color={rgb, 255:red, 74; green, 144; blue, 226 }  ,draw opacity=1 ][fill={rgb, 255:red, 74; green, 144; blue, 226 }  ,fill opacity=1 ] (292.25,32) .. controls (292.25,30.62) and (293.37,29.5) .. (294.75,29.5) .. controls (296.13,29.5) and (297.25,30.62) .. (297.25,32) .. controls (297.25,33.38) and (296.13,34.5) .. (294.75,34.5) .. controls (293.37,34.5) and (292.25,33.38) .. (292.25,32) -- cycle ;
\draw  [color={rgb, 255:red, 74; green, 144; blue, 226 }  ,draw opacity=1 ][fill={rgb, 255:red, 74; green, 144; blue, 226 }  ,fill opacity=1 ] (327.25,62.5) .. controls (327.25,61.12) and (328.37,60) .. (329.75,60) .. controls (331.13,60) and (332.25,61.12) .. (332.25,62.5) .. controls (332.25,63.88) and (331.13,65) .. (329.75,65) .. controls (328.37,65) and (327.25,63.88) .. (327.25,62.5) -- cycle ;
\draw [color={rgb, 255:red, 144; green, 19; blue, 254 }  ,draw opacity=1 ]   (302.5,30) .. controls (324.7,18.9) and (329.79,38.6) .. (330.43,51.12) ;
\draw [shift={(330.5,54)}, rotate = 270] [fill={rgb, 255:red, 144; green, 19; blue, 254 }  ,fill opacity=1 ][line width=0.08]  [draw opacity=0] (10.72,-5.15) -- (0,0) -- (10.72,5.15) -- (7.12,0) -- cycle    ;

\draw (60.58,75.73) node [anchor=north west][inner sep=0.75pt]    {$a$};
\draw (130.58,76.73) node [anchor=north west][inner sep=0.75pt]    {$b$};
\draw (264.08,37.4) node [anchor=north west][inner sep=0.75pt]    {$A( z)$};
\draw (294.25,66.4) node [anchor=north west][inner sep=0.75pt]    {$B( w)$};
\draw (82.58,113.73) node [anchor=north west][inner sep=0.75pt]  [color={rgb, 255:red, 144; green, 19; blue, 254 }  ,opacity=1 ]  {$a\cdot b$};
\draw (215.75,97.4) node [anchor=north west][inner sep=0.75pt]  [color={rgb, 255:red, 144; green, 19; blue, 254 }  ,opacity=1 ]  {$A( z) B( w) \sim \sum _{n}\frac{( A_{( n)} B)( w)}{( z-w)^{n+1}}$};
\draw (359,20.4) node [anchor=north west][inner sep=0.75pt]    {$\bC$};
\end{tikzpicture}
\eea
On a plane, the ``product'' (binary operation) depends on the location holomorphically, leading to infinitely many binary operations.

\begin{defn}
A \textbf{vertex algebra} is a collection of data:
\begin{itemize}
    \item (space of states) a $\bZ$-graded superspace $\cV=\cV_{even}\oplus \cV_{odd}$,
    \item (vacuum) a vector $\ket{0}\in \cV_{even}$,
    \item (translation operator) an even linear map $T: \cV\to\cV$,
    \item (state-field correspondence) an even linear operation (vertex operation) 
    \bea Y(-,z):\cV\to \on{End}\cV[[z,z^{-1}]], \quad A\mapsto Y(A,z)=\sum_{n\in\bZ} A_{(n)}z^{-n-1}\eea
    such that $Y(A,z)B\in \cV((z))$ for any $A,B\in \cV$. 
\end{itemize}
\end{defn}

The data are required to satisfy the following axioms:
\begin{itemize}
    \item (vacuum axiom) $Y(\ket{0},z)=1_{\cV}$, i.e. for any $A\in\cV$, 
    \bea Y(A,z)\ket{0}\in \cV[[z]] \quad \text{and} \quad \lim_{z\to0} Y(A,z)\ket{0}=A,\eea
    \item (translation axiom) $T\ket{0}=0$, i.e. for any $A\in\cV$, 
    \bea \lsb T,Y(A,z)\rsb=\p_z Y(A,z),\eea
    \item (locality axiom) all $\lcb Y(A,z)\rcb_{a\in\cV}$ are mutually local.
\end{itemize}

Roughly speaking, mutual locality implies for any $A,B\in\cV$, we can expand as
\bea Y(A,z)Y(B,w)=\sum_{n\in\bZ}\frac{Y(A_{(n)}\cdot B,w)}{(z-w)^{n+1}}.\eea
This is called the \textbf{operator product expansion (OPE)}. $\lcb A_{(n)}\cdot B\rcb$ from the expansion coefficient can be viewed as defining an infinite tower of products. For simplicity, we will write
\bea A(z)\equiv Y(A,z) \quad \text{for }A\in\cV.\eea
Then the OPE can be written as
\bea A(z)B(w)=\sum_{n\in\bZ}\frac{A_{(n)}\cdot B(w)}{(z-w)^{n+1}}.\eea
We also write, whenever only the \emph{singular} parts matter, 
\bea A(z)B(w)\sim \sum_{n\geq 0}
\frac{A_{(n)}\cdot B(w)}{(z-w)^{n+1}}.\eea

Given a vertex algebra, we can define its \textbf{modes Lie algebra} 
\bea \oint \cV\coloneqq \on{Span}_{\bC}\lcb \oint dz\ z^k A(z)=A_{(k)}\rcb_{A\in\cV,\ k\in\bZ}.\eea
The Lie bracket of contour integrals is determined by the OPE,
\bea \lsb \oint dz\ z^m A(z),\ 
\oint dw\ w^n B(w)\rsb
=\oint dw\ w^n \oint_w dz\ z^m \sum_{j\in\bZ} \frac{A_{(j)}\cdot B(w)}{(z-w)^{j+1}},\eea
where only the singular part matters in the integration. The Lie bracket is represented diagrammatically as follows.
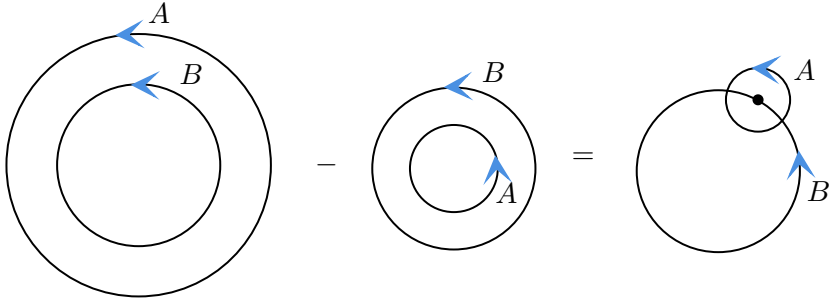
\bea 
\tikzset{every picture/.style={line width=0.75pt}}         
\begin{tikzpicture}[x=0.75pt,y=0.75pt,yscale=-1,xscale=1]

\draw   (49.33,109) .. controls (49.33,72.55) and (78.88,43) .. (115.33,43) .. controls (151.78,43) and (181.33,72.55) .. (181.33,109) .. controls (181.33,145.45) and (151.78,175) .. (115.33,175) .. controls (78.88,175) and (49.33,145.45) .. (49.33,109) -- cycle ;
\draw   (74.61,109) .. controls (74.61,86.51) and (92.84,68.28) .. (115.33,68.28) .. controls (137.82,68.28) and (156.06,86.51) .. (156.06,109) .. controls (156.06,131.49) and (137.82,149.72) .. (115.33,149.72) .. controls (92.84,149.72) and (74.61,131.49) .. (74.61,109) -- cycle ;
\draw  [color={rgb, 255:red, 74; green, 144; blue, 226 }  ,draw opacity=1 ][fill={rgb, 255:red, 74; green, 144; blue, 226 }  ,fill opacity=1 ] (115.78,48.51) -- (104.68,43.53) -- (115.32,37.64) -- (110.11,43.3) -- cycle ;
\draw   (250.75,110.67) .. controls (250.75,98.56) and (260.56,88.75) .. (272.67,88.75) .. controls (284.77,88.75) and (294.58,98.56) .. (294.58,110.67) .. controls (294.58,122.77) and (284.77,132.58) .. (272.67,132.58) .. controls (260.56,132.58) and (250.75,122.77) .. (250.75,110.67) -- cycle ;
\draw   (231.94,110.67) .. controls (231.94,88.18) and (250.18,69.94) .. (272.67,69.94) .. controls (295.16,69.94) and (313.39,88.18) .. (313.39,110.67) .. controls (313.39,133.16) and (295.16,151.39) .. (272.67,151.39) .. controls (250.18,151.39) and (231.94,133.16) .. (231.94,110.67) -- cycle ;
\draw   (408.48,76.1) .. controls (408.48,67.25) and (415.65,60.08) .. (424.5,60.08) .. controls (433.34,60.08) and (440.51,67.25) .. (440.51,76.1) .. controls (440.51,84.94) and (433.34,92.11) .. (424.5,92.11) .. controls (415.65,92.11) and (408.48,84.94) .. (408.48,76.1) -- cycle ;
\draw   (363.94,112) .. controls (363.94,89.51) and (382.18,71.28) .. (404.67,71.28) .. controls (427.16,71.28) and (445.39,89.51) .. (445.39,112) .. controls (445.39,134.49) and (427.16,152.72) .. (404.67,152.72) .. controls (382.18,152.72) and (363.94,134.49) .. (363.94,112) -- cycle ;
\draw  [fill={rgb, 255:red, 0; green, 0; blue, 0 }  ,fill opacity=1 ] (422.23,76.1) .. controls (422.23,74.84) and (423.24,73.83) .. (424.5,73.83) .. controls (425.75,73.83) and (426.76,74.84) .. (426.76,76.1) .. controls (426.76,77.35) and (425.75,78.36) .. (424.5,78.36) .. controls (423.24,78.36) and (422.23,77.35) .. (422.23,76.1) -- cycle ;
\draw  [color={rgb, 255:red, 74; green, 144; blue, 226 }  ,draw opacity=1 ][fill={rgb, 255:red, 74; green, 144; blue, 226 }  ,fill opacity=1 ] (123.26,74.24) -- (112.47,68.61) -- (123.45,63.36) -- (117.91,68.7) -- cycle ;
\draw  [color={rgb, 255:red, 74; green, 144; blue, 226 }  ,draw opacity=1 ][fill={rgb, 255:red, 74; green, 144; blue, 226 }  ,fill opacity=1 ] (279.78,74.71) -- (268.68,69.73) -- (279.32,63.84) -- (274.11,69.5) -- cycle ;
\draw  [color={rgb, 255:red, 74; green, 144; blue, 226 }  ,draw opacity=1 ][fill={rgb, 255:red, 74; green, 144; blue, 226 }  ,fill opacity=1 ] (432.75,66.32) -- (422.29,60.12) -- (433.53,55.46) -- (427.71,60.5) -- cycle ;
\draw  [color={rgb, 255:red, 74; green, 144; blue, 226 }  ,draw opacity=1 ][fill={rgb, 255:red, 74; green, 144; blue, 226 }  ,fill opacity=1 ] (440.35,113.98) -- (444.86,102.68) -- (451.19,113.07) -- (445.31,108.1) -- cycle ;
\draw  [color={rgb, 255:red, 74; green, 144; blue, 226 }  ,draw opacity=1 ][fill={rgb, 255:red, 74; green, 144; blue, 226 }  ,fill opacity=1 ] (288.95,116.38) -- (293.46,105.08) -- (299.79,115.47) -- (293.91,110.5) -- cycle ;

\draw (118.53,26) node [anchor=north west][inner sep=0.75pt]    {$A$};
\draw (134.27,56.6) node [anchor=north west][inner sep=0.75pt]    {$B$};
\draw (291.87,116.68) node [anchor=north west][inner sep=0.75pt]    {$A$};
\draw (285.2,55.87) node [anchor=north west][inner sep=0.75pt]    {$B$};
\draw (202,102.4) node [anchor=north west][inner sep=0.75pt]    {$-$};
\draw (330,101.07) node [anchor=north west][inner sep=0.75pt]    {$=$};
\draw (440.67,54.73) node [anchor=north west][inner sep=0.75pt]    {$A$};
\draw (447.39,115.4) node [anchor=north west][inner sep=0.75pt]    {$B$};
\end{tikzpicture}
\eea

We refer to \cite{Kac, Frenkel-Ben} for a systematic discussion on vertex algebras. 

\begin{eg}[$\beta\gamma$-system]
The $\beta\gamma$-system is generated by two \emph{bosonic} fields $\beta(z), \gamma(z)$ with OPE
\bea \beta(z)\gamma(w)\sim \frac{\hbar}{z-w}\sim -\gamma(z)\beta(w).\eea
The vertex algebra $\cV$ is identified with the differential ring
\bea\cV= \nord{\bC[[\p^i\beta,\p^i\gamma]]} [[\hbar]],\eea
where $\nord{\ }$ is the \emph{normal ordering operator}. The general OPE is obtained via \textbf{Wick contractions}. For example,
\bea \nord{\beta(z)\gamma(z)} \nord{\beta(w)\gamma(w)}
&= \underbrace{\frac{\hbar}{z-w} \nord{\gamma(z)\beta(w)}- 
\frac{\hbar}{z-w} \nord{\beta(z)\gamma(w)}}_{\text{1 contraction}}
-\underbrace{\lb\frac{\hbar}{z-w}\rb^2}_{\text{2 contractions}}\\
&=\sum_{k\geq 0}\frac{\hbar}{z-w}\frac{(z-w)^k}{k!} \nord{\p^k\gamma(w)\beta(w)-\p^k\beta(w)\gamma(w)} -\frac{\hbar^2}{(z-w)^2}.\eea
\end{eg}

\begin{eg}[$bc$-system]
The $bc$-system is generated by two \emph{fermionic} fields $b(z), c(z)$ with OPE
\bea b(z)c(w)\sim \frac{\hbar}{z-w}\sim c(z)b(w).\eea
The vertex algebra $\cV$ is identified with the differential ring
\bea\cV= \nord{\bC[[\p^i b,\p^i c]]} [[\hbar]].\eea
The general OPE is generated in the similar way as the $\beta\gamma$-system (but we need to take care of the signs). 
\end{eg}

More generally, we can define a general $\beta\gamma-bc$ system by considering a $\bZ_2$-graded space 
\bea h=h_0\oplus h_1\eea 
with an even symplectic pairing
\bea \lan -,-\ran: \asym^2 h\to \bC.\eea

Let $\lcb a_i\rcb$ be a basis of $h$, then we can define a vertex algebra $\cV_h$ by
\bea\cV_h= \nord{\bC[[\p^k a_i]]} [[\hbar]].\eea
The OPE is generated by
\bea a_i(z)a_j(w)\sim \frac{\hbar}{z-w} \lan a_i,a_j\ran.\eea
In particular, $h_0$ represents the copies of $\beta\gamma$-system; $h_1$ represents the copies of $bc$-system. 

\subsection{Chiral Deformation of $\beta\gamma-bc$ Systems}
We consider the following data:
\begin{itemize}
    \item an elliptic curve $E$ (topologically a torus $T^2$) with linear coordinate $z$ such that $z\sim z+1\sim z+\tau$,
    \item a graded symplectic space $h=h_0\oplus h_1$ with an even symplectic pairing $\lan -,-\ran$.
\end{itemize}
This defines a field theory in BV formalism where the space of fields is
$$
\cE=\Omega^{0,\blt}(E)\otimes h
$$
with $(-1)$-symplectic pair by 
$$
\omega(\varphi_1,\varphi_2)=\int_E dz  \lan \varphi_1,\varphi_2\ran, \quad \varphi_i\in\cE. 
$$
Note that $\omega$ has $\on{deg}=-1$ since we need exactly one $\ols{dz}$ from $\varphi_1,\varphi_2$ to be integrated.

The free theory is given by 
\bea \hf\int_E dz\lan \varphi
,\pb \varphi\ran, \quad \varphi\in\cE.\eea
The local quantum observables form exactly $\beta\gamma-bc$ system. The propagator is given by the {Szeg\"{o} kernel}
\bea \pb^{-1}\sim \frac{1}{z-w}+\text{regular}.\eea

We would like to consider a general interacting theory by turning on \textbf{chiral deformations} of the form
\bea \int \cL\lb \varphi,\p_z\varphi,\p_z^2\varphi,\cdots\rb\eea
which involves only \emph{holomorphic} derivatives. This is related to the vertex algebra
\bea \cV_{h^\vee}= \bC[[\p^i h^\vee]] [[\hbar]]\eea
as follows. Define a map
\bea I: \cV_{h^\vee}\to \sO_{loc}(\cE), \quad \gamma\mapsto I_\gamma.\eea

Explicitly, if $\gamma=\sum \p^{k_1}a_1\cdots \p^{k_m}a_m$, then \bea I_\gamma(\varphi)=i\int_E dz \sum\pm \p^{k_1}_z a_1(\varphi)\cdots
\p^{k_m}_z a_m(\varphi).\eea
Here $a_i\in h^\vee$ and $a_i(\varphi)\in \omega^{0,\blt}(E)$.

\begin{thm}[\cite{LS-VOA}]
For any $\gamma\in \cV_{h^\vee}$, the chiral deformed theory
\bea \hf \int_E dz\lan \varphi,\pb \varphi\ran +I_\gamma(\varphi)\eea
is \textbf{UV finite} in the sense that the limit
\bea e^{\frac{1}{\hbar}I_\gamma[L]}=\lim_{\varepsilon\to 0} e^{\hbar \partial_{P^L_\varepsilon}} e^{\frac{1}{\hbar}I_\gamma}\quad \text{exists.}\eea
\end{thm}

\begin{rmk}
The proof of the UV finiteness theorem is a bit technical.  The reason is different from topological QM, where we saw that the propagator is bounded (although not continuous). Here the graph integral is NOT absolutely convergent. See Section \ref{sec:regularized-integral} for another geometric interpretation \cite{LZ-1} of this fact.
\end{rmk}

Once we have a well-defined $I_\gamma[L]$ described above, we can formulate the {effective QME}
\bea \pb I_\gamma[L]+\hbar\Delta_L I_\gamma[L]+\hf\lcb I_\gamma[L],I_\gamma[L]\rcb_L=0\eea
and ask for the condition of $\gamma$ to satisfy the equation. It turns out that the answer is very simple.

\begin{thm}[\cite{LS-VOA}]
Consider $\gamma\in \cV_{h^\vee}$ and the effective functional $I_\gamma[L]$ defined above via the UV finiteness. Then $I_\gamma[L]$ satisfies the effective QME
\bea \pb I_\gamma[L]+\hbar\Delta_L I_\gamma[L]+\hf\lcb I_\gamma[L],I_\gamma[L]\rcb_L=0\eea
if and only if 
\bea \lsb \oint\gamma,\oint\gamma\rsb=0 \ \in \oint \cV.\eea
\end{thm}

\begin{rmk}
The local quantum observable of the chiral deformed theory is the vertex algebra $H^\blt(\cV_{h^\vee},[\oint \gamma,-])$. So $[\oint \gamma,-]$ plays the role of BRST reduction. Reversing this reasoning, vertex algebras coming from the BRST reduction of free field realizations can be realized via the model of chiral deformations above.
\end{rmk}

The above theorem can be glued for a \emph{chiral $\sigma$-model}
\bea \varphi:E\to X\eea
which produces a bundle $\cV(X)\to X$ of chiral vertex algebras on $X$. Then the solution of effective QME asks for a flat connection on $\cV(X)$ of the form
\bea D=d+\frac{1}{\hbar}\lsb \oint\gamma,-\rsb, \ \text{such that}\ D^2=0.\eea
Here $\gamma\in\Omega^1\lb X,\cV(X)\rb$ and $\oint \gamma$ is fiberwise chiral mode operator. This can be viewed as the \emph{chiral analogue of Fedosov connection}.

\subsection{Regularized Integral and UV Finiteness}\label{sec:regularized-integral}
The propagator $\pb^{-1}$ is given by the Szeg\"{o} kernel which exhibits holomorphic poles $\frac{1}{z-w}$ along the diagonal. In general, the Feynman diagram involves $\int_{\Sigma^n} \Omega$, where $\Omega$ exhibits holomorphic poles of arbitrary order when $z_i\to z_j$.
It turns out that such looking divergent integral has an intrinsic \emph{regularization} via its conformal structure. 

For simplicity, we start by considering such an integral $\int_\Sigma \omega$.
Here $\Sigma$ is a Riemann surface, possibly with boundary $\p\Sigma$, $\omega$ is a 2-from on $\Sigma$ with meromorphic poles of arbitrary order along a finite set $D\subset\Sigma$ and $D\cap\p\Sigma=\varnothing$.
\bea 
\tikzset{every picture/.style={line width=0.75pt}} 
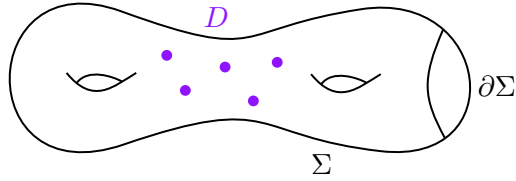
\begin{tikzpicture}[x=0.75pt,y=0.75pt,yscale=-1,xscale=1]

\draw   (79.37,25.97) .. controls (152.1,52.53) and (134.73,34.72) .. (199.52,24.9) .. controls (264.3,15.09) and (274.94,103.85) .. (209.22,96.35) .. controls (143.5,88.85) and (157.79,66.44) .. (80.63,93.01) .. controls (3.47,119.57) and (6.64,-0.59) .. (79.37,25.97) -- cycle ;
\draw    (52.17,56.96) .. controls (65.45,73.4) and (75.57,63.91) .. (86.95,56.96) ;
\draw    (57.23,62.02) .. controls (62.92,54.43) and (75.57,58.86) .. (80,61.39) ;
\draw    (174.23,57.59) .. controls (187.51,74.03) and (197.63,64.55) .. (209.01,57.59) ;
\draw    (179.29,62.65) .. controls (184.98,55.06) and (197.63,59.49) .. (202.06,62.02) ;
\draw  [color={rgb, 255:red, 144; green, 19; blue, 254 }  ,draw opacity=1 ][fill={rgb, 255:red, 144; green, 19; blue, 254 }  ,fill opacity=1 ] (99.98,48.13) .. controls (99.98,46.91) and (100.97,45.92) .. (102.19,45.92) .. controls (103.41,45.92) and (104.4,46.91) .. (104.4,48.13) .. controls (104.4,49.35) and (103.41,50.34) .. (102.19,50.34) .. controls (100.97,50.34) and (99.98,49.35) .. (99.98,48.13) -- cycle ;
\draw  [color={rgb, 255:red, 144; green, 19; blue, 254 }  ,draw opacity=1 ][fill={rgb, 255:red, 144; green, 19; blue, 254 }  ,fill opacity=1 ] (109.25,65.77) .. controls (109.25,64.55) and (110.23,63.57) .. (111.45,63.57) .. controls (112.67,63.57) and (113.66,64.55) .. (113.66,65.77) .. controls (113.66,66.99) and (112.67,67.98) .. (111.45,67.98) .. controls (110.23,67.98) and (109.25,66.99) .. (109.25,65.77) -- cycle ;
\draw  [color={rgb, 255:red, 144; green, 19; blue, 254 }  ,draw opacity=1 ][fill={rgb, 255:red, 144; green, 19; blue, 254 }  ,fill opacity=1 ] (129.09,54.01) .. controls (129.09,52.79) and (130.08,51.81) .. (131.3,51.81) .. controls (132.52,51.81) and (133.5,52.79) .. (133.5,54.01) .. controls (133.5,55.23) and (132.52,56.22) .. (131.3,56.22) .. controls (130.08,56.22) and (129.09,55.23) .. (129.09,54.01) -- cycle ;
\draw  [color={rgb, 255:red, 144; green, 19; blue, 254 }  ,draw opacity=1 ][fill={rgb, 255:red, 144; green, 19; blue, 254 }  ,fill opacity=1 ] (155.12,51.66) .. controls (155.12,50.44) and (156.1,49.45) .. (157.32,49.45) .. controls (158.54,49.45) and (159.53,50.44) .. (159.53,51.66) .. controls (159.53,52.88) and (158.54,53.86) .. (157.32,53.86) .. controls (156.1,53.86) and (155.12,52.88) .. (155.12,51.66) -- cycle ;
\draw    (240.29,34.8) .. controls (229.12,60.68) and (230.29,69.5) .. (240.88,89.49) ;
\draw  [color={rgb, 255:red, 144; green, 19; blue, 254 }  ,draw opacity=1 ][fill={rgb, 255:red, 144; green, 19; blue, 254 }  ,fill opacity=1 ] (143.21,71.07) .. controls (143.21,69.85) and (144.2,68.86) .. (145.41,68.86) .. controls (146.63,68.86) and (147.62,69.85) .. (147.62,71.07) .. controls (147.62,72.28) and (146.63,73.27) .. (145.41,73.27) .. controls (144.2,73.27) and (143.21,72.28) .. (143.21,71.07) -- cycle ;

\draw (172.73,96.1) node [anchor=north west][inner sep=0.75pt]    {$\Sigma $};
\draw (255.91,56.85) node [anchor=north west][inner sep=0.75pt]    {$\partial \Sigma $};
\draw (119.18,22.22) node [anchor=north west][inner sep=0.75pt]  [color={rgb, 255:red, 144; green, 19; blue, 254 }  ,opacity=1 ]  {$D$};
\end{tikzpicture}
\eea

Let $p\in D$ and $z$ be a local coordinate centered at $p$. Then locally $\omega$ can be written as
\bea \omega=\frac{\eta}{z^n}\eea
where $\eta$ is smooth, and $n\in \bZ$.
Since the pole order can be arbitrarily large, the naive integral $\int_\Sigma\omega$ is divergent in general. One homological way out of this divergence problem \cite{LZ-1} is as follows.
We can decompose $\omega$ into
\bea \omega=\alpha+\p\beta,\eea
where $\alpha$ is a 2-form with at most \emph{logarithmic pole} along $D$, $\beta$ is a $(0,1)$-form with \emph{arbitrary order of poles} along $D$, and $\p=dz\frac{\p}{\p z}$ is the holomorphic de Rham differential. Such a decomposition \emph{exists} and is \emph{not unique}.

\begin{defn}[\cite{LZ-1}]
Define the \textbf{regularized integral} 
\bea\dint_\Sigma\omega\coloneqq\int_\Sigma\alpha+\int_{\p\Sigma}\beta\eea
as a recipe to integrate the singular form $\omega$ on $\Sigma$. It has the following properties
\begin{itemize}
    \item it does not depend on the choice of $\alpha,\beta$, and is equivalent to the Cauchy principal value, 
    \item $\dint_\Sigma$ is invariant under conformal transformations,
    \item $\dint_\Sigma\p(-)=\int_{\p\Sigma}(-)$ on $(0,1)$-form with meromorphic poles,
    \item $\dint_\Sigma\pb(-)=\on{Res}(-)$  on $(1,0)$-form with meromorphic poles.
\end{itemize}
\end{defn}
The regularized integral extends the usual integral for smooth forms, i.e., the following diagram is commutative:
\bea \begin{tikzcd}
\cA^2(\Sigma) \ar[rr, hook] \ar[dr, swap, "\int_\Sigma"]
&  & \cA^2(\Sigma,\star D) \ar[dl, "\dint_\Sigma"]\\
& \bC & 
\end{tikzcd}\eea
Here $\cA^2(\Sigma)$ means smooth 2-forms on $\Sigma$, and $\cA^2(\Sigma,\star D)$ means smooth 2-forms on $\Sigma-D$ with meromorphic poles of arbitrary order around $D$. 

We can use this to define integrals on configuration space of $\Sigma$
\bea \on{Conf}_n(\Sigma)=\Sigma^n-\Delta =\lcb\left. (p_1,\cdots,p_n)\in\Sigma^n \right| p_i \neq p_j, \forall i\neq j\rcb\eea
and define
\bea\dint_{\Sigma^n} : \cA^{2n}(\Sigma^n,\star\Delta)\to \bC\eea
by iterating 
\bea \dint_{\Sigma^n}(-)= \dint_{\Sigma}\dint_{\Sigma}\cdots\ \dint_{\Sigma}(-).\eea
It does NOT depend on the choice of the ordering of the factors in $\Sigma^n$: Fubini-type theorem holds.
This gives an intrinsically regularized meaning for $\dint_{\Sigma^n} \Omega$, where $\Omega$ is the Feynman diagram integrand. 
This explains why the theory is UV finite.

\subsection{Chiral Homology and Elliptic Trace}
Intuitively, chiral chain complex can be viewed as a 2d chiral analogue of Hochschild chain complex.
\bea 
\tikzset{every picture/.style={line width=0.75pt}}
\begin{tikzpicture}[x=0.75pt,y=0.75pt,yscale=-1,xscale=1]

\draw    (99.3,113.01) .. controls (114.35,131.65) and (125.82,120.89) .. (138.73,113.01) ;
\draw    (105.03,118.74) .. controls (111.49,110.14) and (125.82,115.16) .. (130.84,118.03) ;
\draw   (60.67,117.78) .. controls (60.67,99.31) and (86.87,84.33) .. (119.19,84.33) .. controls (151.52,84.33) and (177.72,99.31) .. (177.72,117.78) .. controls (177.72,136.25) and (151.52,151.22) .. (119.19,151.22) .. controls (86.87,151.22) and (60.67,136.25) .. (60.67,117.78) -- cycle ;
\draw    (275.96,111.67) .. controls (291.02,130.31) and (302.49,119.56) .. (315.4,111.67) ;
\draw    (281.7,117.41) .. controls (288.15,108.81) and (302.49,113.82) .. (307.51,116.69) ;
\draw   (237.33,116.44) .. controls (237.33,97.97) and (263.54,83) .. (295.86,83) .. controls (328.19,83) and (354.39,97.97) .. (354.39,116.44) .. controls (354.39,134.92) and (328.19,149.89) .. (295.86,149.89) .. controls (263.54,149.89) and (237.33,134.92) .. (237.33,116.44) -- cycle ;

\draw (87.56,93.07) node [anchor=north west][inner sep=0.75pt]    {$\times $};
\draw (116.89,91.07) node [anchor=north west][inner sep=0.75pt]    {$\times $};
\draw (86.89,125.73) node [anchor=north west][inner sep=0.75pt]    {$\times $};
\draw (122.89,126.4) node [anchor=north west][inner sep=0.75pt]    {$\times $};
\draw (137.51,121.43) node [anchor=north west][inner sep=0.75pt]    {$\times $};
\draw (34.89,108.07) node [anchor=north west][inner sep=0.75pt]    {$d_{\on{ch}}$};
\draw (264.22,91.73) node [anchor=north west][inner sep=0.75pt]    {$\times $};
\draw (293.56,89.73) node [anchor=north west][inner sep=0.75pt]    {$\times $};
\draw (263.56,124.4) node [anchor=north west][inner sep=0.75pt]    {$\times $};
\draw (299.56,125.07) node [anchor=north west][inner sep=0.75pt]  [color={rgb, 255:red, 144; green, 19; blue, 254 }  ,opacity=1 ]  {$\times $};
\draw (314.18,120.09) node [anchor=north west][inner sep=0.75pt]  [color={rgb, 255:red, 144; green, 19; blue, 254 }  ,opacity=1 ]  {$\times $};
\draw (184.67,108.07) node [anchor=north west][inner sep=0.75pt]    {$=$};
\draw (208.67,102.07) node [anchor=north west][inner sep=0.75pt]    {$\sum $};
\draw (295.31,134.74) node [anchor=north west][inner sep=0.75pt]  [color={rgb, 255:red, 144; green, 19; blue, 254 }  ,opacity=1 ,rotate=-347.56]  {$\underbrace{\ \ \ \ \ \ }_{\textbf{OPE}}$};
\end{tikzpicture}
\eea

\begin{itemize}
    \item In \cite{zhu1994global}, Zhu studied the space of genus 1 conformal blocks (i.e. the 0th elliptic chiral homology).
    \item In \cite{beilinson2004chiral}, Beilinson and Drinfeld developed the chiral homology theory on general algebraic curves.
\end{itemize}

\paragraph{The construction of Beilinson-Drinfeld.} We briefly review the construction of Beilinson-Drinfeld and refer to \cite{Gui:2021dci} for further details related to the purpose of the current discussion. Let $S$ denote the category of finite non-empty sets whose morphisms are surjections. 
Given the following data:
\begin{itemize}
    \item a category of right $\cD$-modules $\cM(X)$ on $X=\Sigma$,
    \item a category of right $\cD$-modules $\cM(X^S)$ on $X^S$, such that each element $M \in\cM(X^S)$ is a collection that assigns every finite index set $I\in S$ a right $\cD$-module $\cM_{X^I}$ on the product $X^I$ satisfying certain compatibility conditions,
    \item there is an exact fully faithful embedding
    \bea \Delta^{(S)}_\star: \cM(X)\hookrightarrow \cM(X^S)\eea
    via the diagonal map $\Delta^{(I)}:X \hookrightarrow X^I$,
    \item $\cM(X^S)$ carries a (chiral) tensor structure $\otimes^{\on{ch}}$,
\end{itemize}
Then a chiral algebra $\cA$ is a \textbf{Lie algebraic object} via $\Delta^{(S)}_\star$. 

\begin{rmk}
The chiral algebra $\cA$ collects all ``normal ordered operators'' in physics terminology. 
\end{rmk}

We consider the Chevalley-Eilenberg (CE) complex
\bea \lb C(\cA),d_{\on{CE}}\rb=\lb \bigoplus_{\blt>0} \sym^\blt_{\otimes^{\on{ch}}}\lb \Delta^{(S)}_\star\cA[1]\rb, d_{\on{CE}}\rb.\eea
The chiral homology for this complex is
\bea C^{\on{ch}}(X,\cA)=R\Gamma_{DR}(X^S, C(\cA)).\eea

We will focus on $\beta\gamma-bc$ system, where the vertex operator algebra (VOA) $\cV^{\beta\gamma-bc}$ gives rise to a chiral algebra $\cA^{\beta\gamma-bc}$. The following theorem gives the corresponding elliptic trace map in terms of renormalization group flow. 
\begin{thm}[\cite{Gui:2021dci}]\label{thm:Gui-L}
Let $E$ be an elliptic curve. Then the HRG flow gives a map
\bea \lan -\ran_{2d}: C^{\on{ch}}(E,\cA^{\beta\gamma-bc})\to \mathcal O_{\mathrm{BV}}((\hbar))\eea
satisfying the QME $$
(d_{\on{ch}}+\hbar\Delta)\lan -\ran_{2d}=0.
$$
\begin{itemize}

\item Here  $\mathcal O_{\mathrm{BV}}$ is the space of functions on {zero modes} of the $\beta\gamma-bc$ system, which carries a structure of BV algebra. $\Delta$ is the corresponding BV operator. 

\item $\lan -\ran_{2d}$ is defined by
\bea \lan \cO_1\otimes\cdots\otimes\cO_n\ran_{2d}\coloneqq \dint_{E^n} \lb \prod_{i=1}^n \frac{d^2z_i}{\Im \tau}\rb  \lan \cO_1(z_1)\cdots \cO_n(z_n) \ran\eea
where $\abracket{\OO_1(z_1)\cdots \OO_n(z_n)}$ is the correlation function computed via Feynman diagrammatics, and $\dashint$ is the {regularized integral}. 

\item The QME says that ${\langle-\rangle_{2d}}$ intertwines the chiral differential of the elliptic chiral chain complex with the BV operator $-\hbar \Delta$ of the zero-mode algebra $\mathcal O_{\mathrm{BV}}((\hbar)))$. Moreover, ${\langle-\rangle_{2d}}$ is shown to be a {quasi-isomorphism}. 

\item The BV trace with universal background leads to \textbf{Witten genus}.
\end{itemize}
\end{thm}

Theorem \ref{thm:Gui-L} establishes the construction of BV quantization and trace map outlined in the introduction. The Witten genus can be viewed as describing an elliptic chiral analogue of the algebraic index. The computation of Witten genus in BV quantization follows essentially from similar arguments in Costello \cite{Costello-Wittengenus} and Gorbounov-Gwilliam-Williams \cite{GGW}. 

\subsection{2d $\to$ 1d Reduction}\label{sec:2d-1d}
We summarize our discussion as follows.
\begin{table}[!htpb]\centering
            \begin{tabular}{c|c}\toprule
            \textbf{1d TQM} & \textbf{2d chiral QFT}\\ \hline
            Associative algebra & Vertex/chiral algebra\\ \hline
            Hochschild homology & Chiral homology\\ \hline 
            QME $(\hbar\Delta+b)\lan -\ran_{1d}=0$ & QME $(\hbar\Delta+d_{\on{ch}})\lan -\ran_{2d}=0$ \\ \hline
            $\lan \cO_1\otimes\cdots\otimes\cO_n\ran_{1d}=\int_{\ols{\on{Conf}_n(S^1)}}$ & $\lan \cO_1\otimes\cdots\otimes\cO_n\ran_{2d}=\dint_{\Sigma^n}$ \\ 
            \bottomrule
            \end{tabular}
\end{table}

In physics, the partition functions/correlation functions on elliptic curves are described by reducing to a quantum mechanical system on $S^1$.
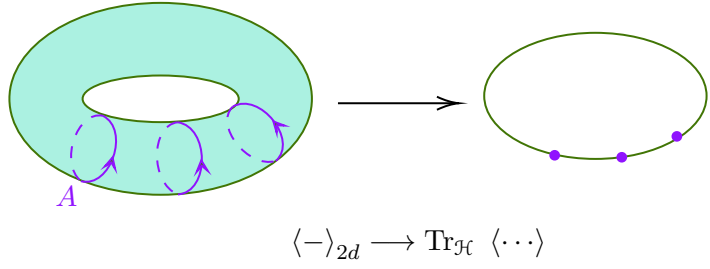
\bea 
\tikzset{every picture/.style={line width=0.75pt}}       
\begin{tikzpicture}[x=0.75pt,y=0.75pt,yscale=-1,xscale=1]

\draw  [color={rgb, 255:red, 65; green, 117; blue, 5 }  ,draw opacity=1 ] (321.5,63.79) .. controls (321.5,46.27) and (346.35,32.07) .. (377,32.07) .. controls (407.65,32.07) and (432.5,46.27) .. (432.5,63.79) .. controls (432.5,81.3) and (407.65,95.5) .. (377,95.5) .. controls (346.35,95.5) and (321.5,81.3) .. (321.5,63.79) -- cycle ;
\draw    (248.5,67.5) -- (306.83,67.5) ;
\draw [shift={(308.83,67.5)}, rotate = 180] [color={rgb, 255:red, 0; green, 0; blue, 0 }  ][line width=0.75]    (10.93,-3.29) .. controls (6.95,-1.4) and (3.31,-0.3) .. (0,0) .. controls (3.31,0.3) and (6.95,1.4) .. (10.93,3.29)   ;
\draw  [color={rgb, 255:red, 144; green, 19; blue, 254 }  ,draw opacity=1 ][fill={rgb, 255:red, 144; green, 19; blue, 254 }  ,fill opacity=1 ] (354.5,93.67) .. controls (354.5,92.47) and (355.47,91.5) .. (356.67,91.5) .. controls (357.86,91.5) and (358.83,92.47) .. (358.83,93.67) .. controls (358.83,94.86) and (357.86,95.83) .. (356.67,95.83) .. controls (355.47,95.83) and (354.5,94.86) .. (354.5,93.67) -- cycle ;
\draw  [color={rgb, 255:red, 144; green, 19; blue, 254 }  ,draw opacity=1 ][fill={rgb, 255:red, 144; green, 19; blue, 254 }  ,fill opacity=1 ] (388,94.67) .. controls (388,93.47) and (388.97,92.5) .. (390.17,92.5) .. controls (391.36,92.5) and (392.33,93.47) .. (392.33,94.67) .. controls (392.33,95.86) and (391.36,96.83) .. (390.17,96.83) .. controls (388.97,96.83) and (388,95.86) .. (388,94.67) -- cycle ;
\draw  [color={rgb, 255:red, 144; green, 19; blue, 254 }  ,draw opacity=1 ][fill={rgb, 255:red, 144; green, 19; blue, 254 }  ,fill opacity=1 ] (415.5,84.17) .. controls (415.5,82.97) and (416.47,82) .. (417.67,82) .. controls (418.86,82) and (419.83,82.97) .. (419.83,84.17) .. controls (419.83,85.36) and (418.86,86.33) .. (417.67,86.33) .. controls (416.47,86.33) and (415.5,85.36) .. (415.5,84.17) -- cycle ;
\draw  [color={rgb, 255:red, 65; green, 117; blue, 5 }  ,draw opacity=1 ][fill={rgb, 255:red, 80; green, 227; blue, 194 }  ,fill opacity=0.48 ,even odd rule] (121.03,65.26) .. controls (121.03,58.79) and (138.48,53.54) .. (160,53.54) .. controls (181.52,53.55) and (198.96,58.8) .. (198.96,65.28) .. controls (198.96,71.75) and (181.52,77) .. (160,77) .. controls (138.48,76.99) and (121.03,71.74) .. (121.03,65.26)(84.5,65.26) .. controls (84.51,38.61) and (118.31,17.01) .. (160.01,17.01) .. controls (201.7,17.02) and (235.5,38.63) .. (235.49,65.28) .. controls (235.49,91.93) and (201.68,113.53) .. (159.99,113.53) .. controls (118.29,113.52) and (84.5,91.91) .. (84.5,65.26) ;
\draw [color={rgb, 255:red, 144; green, 19; blue, 254 }  ,draw opacity=1 ]   (196,69.19) .. controls (209.5,61.53) and (229,86.03) .. (217.5,96.53) ;
\draw [color={rgb, 255:red, 144; green, 19; blue, 254 }  ,draw opacity=1 ] [dash pattern={on 4.5pt off 4.5pt}]  (196,69.19) .. controls (186.5,78.53) and (203.5,100.53) .. (217.5,96.53) ;
\draw  [color={rgb, 255:red, 144; green, 19; blue, 254 }  ,draw opacity=1 ][fill={rgb, 255:red, 65; green, 117; blue, 5 }  ,fill opacity=1 ] (217.89,82.25) -- (216.27,76.66) -- (221.72,78.72) -- (218.04,78.57) -- cycle ;
\draw [color={rgb, 255:red, 144; green, 19; blue, 254 }  ,draw opacity=1 ]   (168.5,77.07) .. controls (184.43,79.43) and (184.7,112.12) .. (168.59,113.38) ;
\draw [color={rgb, 255:red, 144; green, 19; blue, 254 }  ,draw opacity=1 ] [dash pattern={on 4.5pt off 4.5pt}]  (168.5,77.07) .. controls (154.77,78.65) and (154.6,107.68) .. (168.59,113.38) ;
\draw  [color={rgb, 255:red, 144; green, 19; blue, 254 }  ,draw opacity=1 ][fill={rgb, 255:red, 65; green, 117; blue, 5 }  ,fill opacity=1 ] (178.04,101.88) -- (180.29,96.23) -- (183.42,101.43) -- (180.51,98.94) -- cycle ;
\draw [color={rgb, 255:red, 144; green, 19; blue, 254 }  ,draw opacity=1 ]   (131.17,73.82) .. controls (144.86,80.83) and (135.11,110.48) .. (120.16,106.69) ;
\draw [color={rgb, 255:red, 144; green, 19; blue, 254 }  ,draw opacity=1 ] [dash pattern={on 4.5pt off 4.5pt}]  (131.17,73.82) .. controls (118.27,71.06) and (109.24,97.26) .. (120.16,106.69) ;
\draw  [color={rgb, 255:red, 144; green, 19; blue, 254 }  ,draw opacity=1 ][fill={rgb, 255:red, 65; green, 117; blue, 5 }  ,fill opacity=1 ] (132.22,99.18) -- (135.98,94.76) -- (137.22,100.42) -- (135.35,97.28) -- cycle ;

\draw (223.5,128.9) node [anchor=north west][inner sep=0.75pt]    {$\lan - \ran_{2d} \longrightarrow \Tr_{\cH}\ \lan\cdots\ran$};
\draw (106,107.59) node [anchor=north west][inner sep=0.75pt]  [color={rgb, 255:red, 144; green, 19; blue, 254 }  ,opacity=1 ]  {$A$};
\end{tikzpicture}
\eea
Now we can define 2d chiral correlation function using \emph{regularized integral} $\dint_E$. In 1d, operators are described by $A$-cycle $\oint_A$. These two integrals are not exactly the same, but related to each other by \emph{holomorphic limit}.

\begin{thm}[\cite{LZ-1}]
Let $\Phi(z_1,\cdots,z_n;\tau)$ be a meromorphic elliptic function on $\bC^n\times \bH$ which is holomorphic away from diagonals. Let $A_1,\cdots,A_n$ be $n$ disjoint $A$-cycles on the elliptic curve $E_\tau=\bC/(\bZ\oplus \tau\bZ)$. 
\bea 
\tikzset{every picture/.style={line width=0.75pt}} 
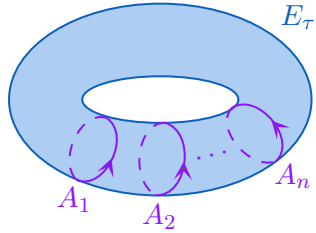
\begin{tikzpicture}[x=0.75pt,y=0.75pt,yscale=-1,xscale=1]
\draw  [color={rgb, 255:red, 11; green, 95; blue, 193 }  ,draw opacity=1 ][fill={rgb, 255:red, 74; green, 144; blue, 226 }  ,fill opacity=0.45 ,even odd rule] (244.53,97.57) .. controls (244.53,91.09) and (261.98,85.84) .. (283.5,85.85) .. controls (305.02,85.85) and (322.46,91.11) .. (322.46,97.58) .. controls (322.46,104.06) and (305.02,109.31) .. (283.5,109.3) .. controls (261.98,109.3) and (244.53,104.05) .. (244.53,97.57)(208,97.56) .. controls (208.01,70.91) and (241.81,49.31) .. (283.51,49.32) .. controls (325.2,49.33) and (359,70.94) .. (358.99,97.59) .. controls (358.99,124.24) and (325.18,145.84) .. (283.49,145.83) .. controls (241.79,145.83) and (208,124.21) .. (208,97.56) ;
\draw [color={rgb, 255:red, 144; green, 19; blue, 254 }  ,draw opacity=1 ]   (319.5,101.5) .. controls (333,93.83) and (352.5,118.33) .. (341,128.83) ;
\draw [color={rgb, 255:red, 144; green, 19; blue, 254 }  ,draw opacity=1 ] [dash pattern={on 4.5pt off 4.5pt}]  (319.5,101.5) .. controls (310,110.83) and (327,132.83) .. (341,128.83) ;
\draw  [color={rgb, 255:red, 144; green, 19; blue, 254 }  ,draw opacity=1 ][fill={rgb, 255:red, 65; green, 117; blue, 5 }  ,fill opacity=1 ] (341.39,114.56) -- (339.77,108.96) -- (345.22,111.03) -- (341.54,110.88) -- cycle ;
\draw [color={rgb, 255:red, 144; green, 19; blue, 254 }  ,draw opacity=1 ]   (285.75,109.23) .. controls (301.43,113.61) and (297.57,146.42) .. (281.26,145.65) ;
\draw [color={rgb, 255:red, 144; green, 19; blue, 254 }  ,draw opacity=1 ] [dash pattern={on 4.5pt off 4.5pt}]  (285.75,109.23) .. controls (271.79,109.09) and (267.95,138.17) .. (281.26,145.65) ;
\draw  [color={rgb, 255:red, 144; green, 19; blue, 254 }  ,draw opacity=1 ][fill={rgb, 255:red, 65; green, 117; blue, 5 }  ,fill opacity=1 ] (292.18,135.31) -- (295.16,129.94) -- (297.64,135.54) -- (295.04,132.68) -- cycle ;
\draw [color={rgb, 255:red, 144; green, 19; blue, 254 }  ,draw opacity=1 ]   (256.75,106.87) .. controls (269.43,115.57) and (255.96,143.72) .. (241.61,138.05) ;
\draw [color={rgb, 255:red, 144; green, 19; blue, 254 }  ,draw opacity=1 ] [dash pattern={on 4.5pt off 4.5pt}]  (256.75,106.87) .. controls (244.31,102.47) and (231.99,127.3) .. (241.61,138.05) ;
\draw  [color={rgb, 255:red, 144; green, 19; blue, 254 }  ,draw opacity=1 ][fill={rgb, 255:red, 65; green, 117; blue, 5 }  ,fill opacity=1 ] (254.53,132.14) -- (258.83,128.25) -- (259.34,134.02) -- (257.88,130.66) -- cycle ;

\draw (230.33,140.23) node [anchor=north west][inner sep=0.75pt]  [color={rgb, 255:red, 144; green, 19; blue, 254 }  ,opacity=1 ]  {$A_{1}$};
\draw (310.43,125.53) node  [color={rgb, 255:red, 144; green, 19; blue, 254 }  ,opacity=1 ,rotate=-346.47]  {$\cdots $};
\draw (339,128.4) node [anchor=north west][inner sep=0.75pt]  [color={rgb, 255:red, 144; green, 19; blue, 254 }  ,opacity=1 ]  {$A_{n}$};
\draw (273,148.9) node [anchor=north west][inner sep=0.75pt]  [color={rgb, 255:red, 144; green, 19; blue, 254 }  ,opacity=1 ]  {$A_{2}$};
\draw (342,48.4) node [anchor=north west][inner sep=0.75pt]  [color={rgb, 255:red, 11; green, 95; blue, 193 }  ,opacity=1 ]  {$E_{\tau }$};
\end{tikzpicture}
\eea
Then the regularized integral 
\bea\dint_{E^n_\tau}\lb \prod_{i=1}^n \frac{d^2z_i}{\Im \tau}\rb \Phi(z_1,\cdots,z_n;\tau)\eea
lies in $\sO_{\bH}\lsb\frac{1}{\Im \tau}\rsb$. Moreover, we have
\bea\lim_{\ols{\tau}\to\infty}\dint_{E^n_\tau} \lb \prod_{i=1}^n\frac{d^2z_i}{\Im \tau}\rb \Phi
=\frac{1}{n!}\sum_{\sigma\in S_n}\oint_{A_{\sigma(1)}} dz_1 \cdots \oint_{A_{\sigma(n)}} dz_n\ \Phi,\eea
where $S_n$ is the $n$-th permutation group and 
$$
\lim_{\ols{\tau}\to\infty}: \sO_{\bH}\lsb\frac{1}{\Im \tau}\rsb\to \sO_{\bH}\quad \text{is the map sending} \quad \frac{1}{\Im \tau}\to 0.
$$
\end{thm}

This theorem gives a precise relation on reduction of torus to circle
\bea\dint_{E^n}\xrightarrow{\lim\limits_{\ols{\tau}\to\infty}} \text{averaged } \oint_A. \eea
The anti-holomorphic dependence of $\dint_{E^n}$ on the moduli $\tau$ is actually fully described by the holomorphic anomaly equation \cite{LZ-2}.

Furthermore, if $\Phi(z_1,\cdots,z_n;\tau)$ is modular of weight $m$, then its regularized integral $\dint_{E^n_\tau}\lb \prod_{i=1}^n \frac{d^2z_i}{\Im \tau}\rb \Phi(z_1,\cdots,z_n;\tau)$ is modular of weight $m$ and thus an almost holomorphic modular forms \cite{Kaneko-Zagier}. The holomorphic limit by averaged $\oint_A$ is a quasi-modular form of weight $m$.

We apply the above theorem to 2d chiral correlations on elliptic curves. This leads to the following relation between the elliptic trace map in Theorem \ref{thm:Gui-L} and Weyl-ordered  operators by $A$-cycle integrals
\begin{align*}
\lim_{\ols{\tau}\to\infty} \lan \cO_1\otimes\cdots\otimes\cO_n\ran_{2d}=&\lim_{\ols{\tau}\to\infty}\dint_{E^n_\tau} \lb \prod_{i=1}^n\frac{d^2z_i}{\Im \tau}\rb \abracket{\OO_1(z_1)\cdots \OO_n(z_n)}\\
=&\frac{1}{n!}\sum_{\sigma\in S_n}\oint_{A_{\sigma(1)}} dz_1 \cdots \oint_{A_{\sigma(n)}} dz_n\ \abracket{\OO_1(z_1)\cdots \OO_n(z_n)}.
\end{align*}
This can be viewed as a reduction formula from 2d to 1d. This formula illustrates an interesting relationship between regularization and modularity/quasi-modularity.

\subsection{Application to Mirror Symmetry}

Mirror symmetry is about a duality between 
$$
{\small \text{\boxed{\text{symplectic geometry}} ({A-model})} \Longleftrightarrow \text{\boxed{\text{complex geometry}} ({B-model)}}}
$$
Here is a cartoon to illustrate how such mirror relation arises from physics.
$$
\xymatrix{
   \int_{\Map(\Sigma_g, X)} \bracket{\text{A-model}}\ar[d]^{\text{SUSY localize}} \ar@{=}[rr]^{\text{Fourier transform}} &  &\int_{\Map(\Sigma_g, X^\prime)} \bracket{\text{B-model}}\ar[d]_{\text{SUSY localize}}\\
 \int_{\text{Holomorphic maps}(\Sigma_g, X)} \ar@{<-->}[rr]  \ar@{=>}[d]  && 
 \int_{\text{Constant maps}(\Sigma_g, X^\prime)}   \ar@{=>}[d]  \\
 \text{Gromov-Witten Theory} && \text{Hodge theory/Kodaira-Spencer gravity}
}
$$

Consider the example of elliptic curves, whose mirrors are elliptic curves as well. The full quantum B-model (quantum BCOV theory as developed in \cite{Costello-Li}) on elliptic curves (including all gravitational descendents) is completely solved in \cite{LS-VOA}. The so-called stationary sector is described by the chiral deformation of chiral boson
$$
S=\int \pa \phi\wedge \bar\partial \phi+\sum_{k\geq 0}\int \eta_k {W^{(k+2)}(\pa_z\phi)\over k+2}
$$
where 
$$
W^{(k)}(\pa_z\phi)= (\pa_z\phi)^k+ O(\hbar)
$$
are the bosonic realization of the $W_{1+\infty}$-algebra.  The holomorphic limit $\bar \tau \to \infty$ (explained in Section \ref{sec:2d-1d}) of the generating function of $S$ on the elliptic curve  coincides  with the Gromov-Witten invariants on the mirror computed by Dijkgraaf\cite{Dijkgraaf-elliptic} and Okounkov-Pandharipande\cite{OP-Virasoro}. In this case, we find \cite{LS-VOA} 
$$
\text{Quantum Mirror Symmetry=Boson-Fermion Correspondence}.
$$

\bigskip{}

\noindent{\small Yau Mathematical Sciences Center, Tsinghua University, Beijing 100084, China}

\noindent{\small Email: \tt  sili@mail.tsinghua.edu.cn}


\end{document}